\documentclass[12pt,a4paper,onecolumn]{ieeecolor}
\input{generic.sty}
\usepackage{cite}
\usepackage{amsmath,amssymb,amsfonts}
\usepackage{algorithmic}
\usepackage{graphicx}
\usepackage[a4paper,margin=1in]{geometry}
\usepackage{textcomp}
\usepackage{siunitx}  
\usepackage{comment}
\usepackage{mathtools}
\usepackage{xcolor}
\usepackage[caption=false]{subfig}
\usepackage{float}
  \usepackage{setspace}
\usepackage{hyperref}
\usepackage{float}
\usepackage{multirow}
\usepackage{cases}
\usepackage{empheq}
\hypersetup{colorlinks=false}
\hypersetup{linkcolor=black, citecolor=black, urlcolor=black}                                              
\usepackage{bm}
\usepackage[utf8]{inputenc}

\def\BibTeX{{\rm B\kern-.05em{\sc i\kern-.025em b}\kern-.08em
    T\kern-.1667em\lower.7ex\hbox{E}\kern-.125emX}}
\begin{document}
\newtheorem{theorem}{\textbf{Theorem}}
\newtheorem{remark}{\textbf{Remark}}
\newtheorem{lemma}{\textbf{Lemma}}
\newtheorem{proposition}{\textbf{Proposition}}
\newtheorem{corollary}{Corollary}

\title{Research note Twente}
\singlespacing

\begin{center}{\large \textbf{ Multi-Period Repetitive Control Design in a Time Delay Framework with Application to an Active Vibration Isolation System}}
\end{center}
\begin{center}
Xiaoran Han\footnote{Xiaoran Han is with Robotics for Extreme Environments, Department of Electrical and Electronic Engineering, University of Manchester, Manchester, M13 9PL, UK (email: xiaoran.han@manchester.ac.uk)}, Sil Spanjer and Wouter B.J. Hakvoort\footnote{Sil Spanjer and Wouter B.J. Hakvoort are with Faculty of Engineering Technology, Precision Engineering, Twente University, Enschede, 7522 NB, Netherlands (email: s.t.spanjer@gmail.com), (email: w.b.j.hakvoort@utwente.nl)}
\end{center}

\section{\textbf{Abstract}}
Multi-period Repetitive Control (MPRC) can lead to amplification at non-repetitive frequencies due to the multiplication of multi-period repetitive control in their closed-loop sensitivity functions, which leads to closed-to imaginary poles. We propose a MPRC structure that eliminates such interaction.
The scheme can be equivalently converted into a disturbance observer to improve robustness to frequency uncertainties. We derive an upper bound on the sensitivity function that can be optimized at either repetitive or non-repetitive frequencies. Experimental results on an active vibration isolation system show superior performance over existing robust MPRC schemes, and the capability to suppress unknown time-varying periodic disturbances.

\section{\textbf{Introduction}}
Repetitive Control (RC), which employs the internal model principle, 
is well known for its superior performance in rejecting periodic disturbances and their harmonics that occur at a single, known fundamental frequency \cite{Tomizuka}, \cite{Hara}, \cite{Tadashi}.  
It has been applied in various areas, such as in motion control systems, particularly in precision positioning applications such as robotics \cite{Kazumasa}, CNC machines, and servo systems \cite{Huixing},
power inverters \cite{Tianqi}, periodic wind rejection \cite{Ivo},  tracking control of a CD-player \cite{Steinbuch}, 
vibration isolation \cite{Daley} or tremor suppression \cite{Bing1}.
RC assumes the system has a relative degree of zero \cite{Steinbuch}, i.e., the output measurements directly depend on the inputs without any differentiation. Based on this,
signals are constructed utilizing both the control and measurement signals that were sampled at a previous $\tau$ period.
 The signals are then fed forward, along with a nominal feedback controller, to the plant that experiences external periodic disturbances. The disturbances can be fully compensated by the signals if the exact value of $\tau$, which is the period of the disturbances, is known in advance \cite{Tadashi}.

In the practical implementation of RC, designers are expected to address the following three main limitations.
First of all,  the exact frequency of disturbance must be known in advance. When the exact frequency is unknown, estimation techniques such as the disturbance observer are required for estimating the unknown periodic disturbance \cite{Tsao}, \cite{Bodson}, \cite{Bodson2}, \cite{Min Wu2}.
Secondly, RC is highly sensitive to frequency uncertainties, even when the disturbance frequency is known. Its performance 
deteriorates rapidly in the presence of even slight variations in disturbance frequency. 
Addressing this, \cite{Steinbuch},   \cite{Inoue}, \cite{Chang1}  proposed a high-order RC scheme for improving RC robustness to frequency uncertainties. The reduced high sensitivity, however, comes at the cost of 
large disturbance amplification at the intermediate frequencies between each harmonic.
 Studies are carried out using linear programming \cite{Steinbuch2},   semi-definite programming \cite{Pipeleers}, $H_\infty$ and LQ design \cite{Min Wu}, \cite{Koroglu},    to yield an optimal trade-off between robustness for changes in the period time and for reduction of the spectrum in-between the harmonic frequencies. 
 Thirdly, unlike the celebrated success of single-period RC in various industrial domains, multi-period RC (MPRC) does not receive as wide acceptance. 
There are mainly two  MPRC schemes available, i.e., parallel structure \cite{W.Chang}, \cite{Owens} or cascaded structure \cite{M.Yamada}, \cite{Arancibia}, \cite{Winarto}, \cite{Blanken}.
Due to the interaction within MPRC,
the parallel structure results in severe amplifications at specific frequencies where harmonics from different RCs constructively interfere. 
In the face of this limitation,  cascaded MPRC design assumes the exact knowledge of the nominal closed-loop dynamics is available or can be learned, and presumes implementing a high nominal loop gain control, or a unity complementary sensitivity function to suppress the interactions.
However, there is a lack of rigorous stability guarantees for implementing high nominal loop gain control and applying high RC gain without violating stability.
 Last but not least, RC has been widely regarded as a plug-in controller, alongside the nominal closed-loop control. 
The controlled signal from the RC applies to both the disturbance and the nominal control loop.
The interaction between the RC and the nominal controller sometimes prevents the use of higher RC gains, potentially rendering MPRC ineffective. 

Simultaneously, at about the same time when RC was originated \cite{Tomizuka}, \cite{Hara}, \cite{Tadashi},  another disturbance compensation scheme, namely Time Delay Control (TDC) \cite{Hsia}, \cite{Youcef}, originated for unknown random disturbances.
 The two schemes share a similar concept in their disturbance compensation techniques, which utilize delays in the control signals. The RC utilizes the past $\tau$ samples of control inputs and measured outputs, where $\tau$ denotes the period of the external disturbance. On the other hand, TDC requires previous samples of the control inputs and the measured outputs, with $t_s$ denoting the sampling period of the digital control system. 
While such equivalent disturbance reconstruction allows a system to be robust to a specified disturbance,
the design of the past control signals in generating the present control signals, however, requires  more attention, as the control signal risks being integrated due to time-delayed compensation
of the disturbance.

In this paper, we propose a new MPRC configuration inspired by the RC and the TDC. It is shown
that the high-order RC proposed in \cite{Steinbuch} is a special case of the proposed MPRC scheme. The inclusion of TDC in RC improves robustness to frequency uncertainties and provides an estimate of unknown periodic disturbances.
The MPRC interactions, which are detrimental in parallel or cascaded RC and exhibit significant amplification at the interacting frequencies, are not present in the proposed approach.  We note the proposed approach does not assume a unity complementary sensitivity function in order to suppress the multi-frequency interactions.
Inspired by the proposed MPRC structure, it is shown that the conventional cascaded MPRC can be
modified to share essentially the same feature of the proposed MPRC but without the plant model being known. 
Delay-dependent input-to-state stability conditions based on Linear Matrix Inequalities (LMIs) are derived, which allow optimization of RC parameters for a given frequency, even if it is not a repetitive frequency.
The proposed RC is experimentally validated on a six-axis, hard-mounted Active Vibration Isolation System to demonstrate its efficiency. 

\textbf{Notation} Through out the paper, time dependent variable $x(t)$ is simply denoted as $x$, otherwise it is denoted as $x(t-\tau)$ if it is delayed by a period of $\tau$. 
 Matrix $P>0$ means it is symmetric and positive definite. 

\section{\textbf{System description}}

Consider the following mass-spring-damper model
\begin{equation}\label{linearSys}
\begin{array}{l}
\dot x_1(t)=x_2(t),\\
m\dot x_2(t)=-kx_1(t)-dx_2(t)+F_1(t)+F_d (\omega_{d_1},\,\ldots,\,\omega_{d_{n-1}},\,\omega_{d_n}),
\end{array}
\end{equation}
where $x_1(t),\, x_2(t),\, \dot x_2(t)\in \mathbb{R}$ denotes  position, velocity and acceleration. We consider the case where the acceleration $\dot x_2$ is measured, as in our setup. The velocity $x_2$ can either be measured or estimated from the acceleration signal using a low-pass filter or a weak integrator with an appropriate cut-off frequency, e.g., 1 rad/s for noise filtering. In our controller formulation, we provide full-state feedback with $x_1$ measured.  For a stable system like (\ref{linearSys}), it is, however, not a necessary condition for the feasibility of our stability condition.  
$m$, $k$, $d$, $F_1(t)$ denote the mass, stiffness, damping, and control input, respectively. $F_d (\omega_{d_1},\,\ldots,\,\omega_{d_{n-1}},\,\omega_{d_n})=\sum_{i=1}^nF_{d_i}(\omega_{d_i}) $ denotes disturbances, composed of multiple fundamental frequencies $\omega_{d_i}$ where each frequency component is bounded by $|F_d(\omega_{d_i})|\leq \Delta_i$, where $\Delta_i\geq 0$.
We consider a scalar system for simplicity in demonstrating our result.
Nevertheless, higher-dimensional systems can also be considered after modal decomposition.  
 As in robotic systems \cite{Kazumasa}, the mass matrix and the inertia tensor depend on the joint angles and velocities; we consider a time-varying or uncertain mass quantity $m=\overline m+\Delta m$, where $\overline m$ and $\Delta m$ denote the known and uncertain parts, respectively. 
 Denoting $\phi_1(t)=-\Delta m\dot x_2(t)-kx_1(t)-dx_2(t)+F_d (\omega_{d_1},\,\ldots,\,\omega_{d_{n-1}},\,\omega_{d_n})$ as a total disturbance, 
 $\dot x_2$ in (\ref{linearSys}) can be rewritten as 
\begin{equation}\label{s1}
\overline m\dot x_2(t)=\phi_1(t)+F_1(t).
\end{equation}
The stiffness $k$ and damping $d$ are denoted as part of the total disturbance $\phi$ for now; however, this does not restrict us from deriving a transfer function of the closed-loop system acceleration $\dot x _ {2} $ with respect to the actual external disturbance $F_d$.
For a system subject to external periodic disturbance,
RC produces 
a control signal that is the sum of the previous $\tau$ sample of the control signal and the previous $\tau$ sample of the output signal that is directly a function of the input without differentiation, i.e., a relative degree zero system \cite{Steinbuch} . $\tau$ denotes the period of the external periodic disturbance whose fundamental frequency is $\omega=\frac{2\pi}{\tau}$.

Inspired by single-period RC, we propose a multi-period RC that is designed by sequentially adding single-period RC controllers, each designed for a disturbance with period $\tau_i$. This leads to
\begin{equation}\label{s2_time}
\begin{array}{l}
F_i(t)=Q_i\lambda_i (F_i(t-\tau_i)-\overline m\dot x_2(t-\tau_i))+F_{i+1}(t), \textrm{\,\,\,for\,\,\,} i=1,\ldots, n-1,\\
F_n(t)=Q_n\lambda_n(F_n (t-\tau_n)-\overline m\dot x_2(t-\tau_n))-px_1(t)-gx_2(t),
\end{array}
\end{equation}
where $Q_i$ is a low-pass or band-pass filter,  $\lambda_i\geq 0$ denotes an RC gain and
   $p,\,g\geq 0$ are some positive nominal control parameters. We acknowledge that translation of the filters in (\ref{s2_time}) into the time domain requires careful numerical implementation due to the combined effects of filtering and time delays. It is possible to implement the filtering using a discrete-time version of the controller, such as in dSPACE, in our setup.

In the frequency domain, this can be expressed as 
\begin{equation}\label{s2}
\begin{array}{l}
F_i(s)=Q_i(s)\lambda_i e^{-s\tau_i}  (F_i(s)-\overline m s X_2(s))+F_{i+1}(s), \textrm{\,\,\,for\,\,\,} i=1,\ldots, n-1,\\
F_n(s)=Q_n(s)\lambda_ne^{-s\tau_n} (F_n (s)-\overline m s X_2(s))-CsX_2(s),
\end{array}
\end{equation}
where $C=\frac{gs+p}{s^2}$ denotes the nominal control function. As a result, a schematic of the nominal system (\ref{linearSys}) with the proposed multi-period RC in (\ref{s2})   is shown in Fig. \ref{fig:2H}, where  $G=\frac{s^2}{ms^2+ds+k}$ denotes the nominal plant.
\begin{figure}
\centering
  \includegraphics[width=1\linewidth]{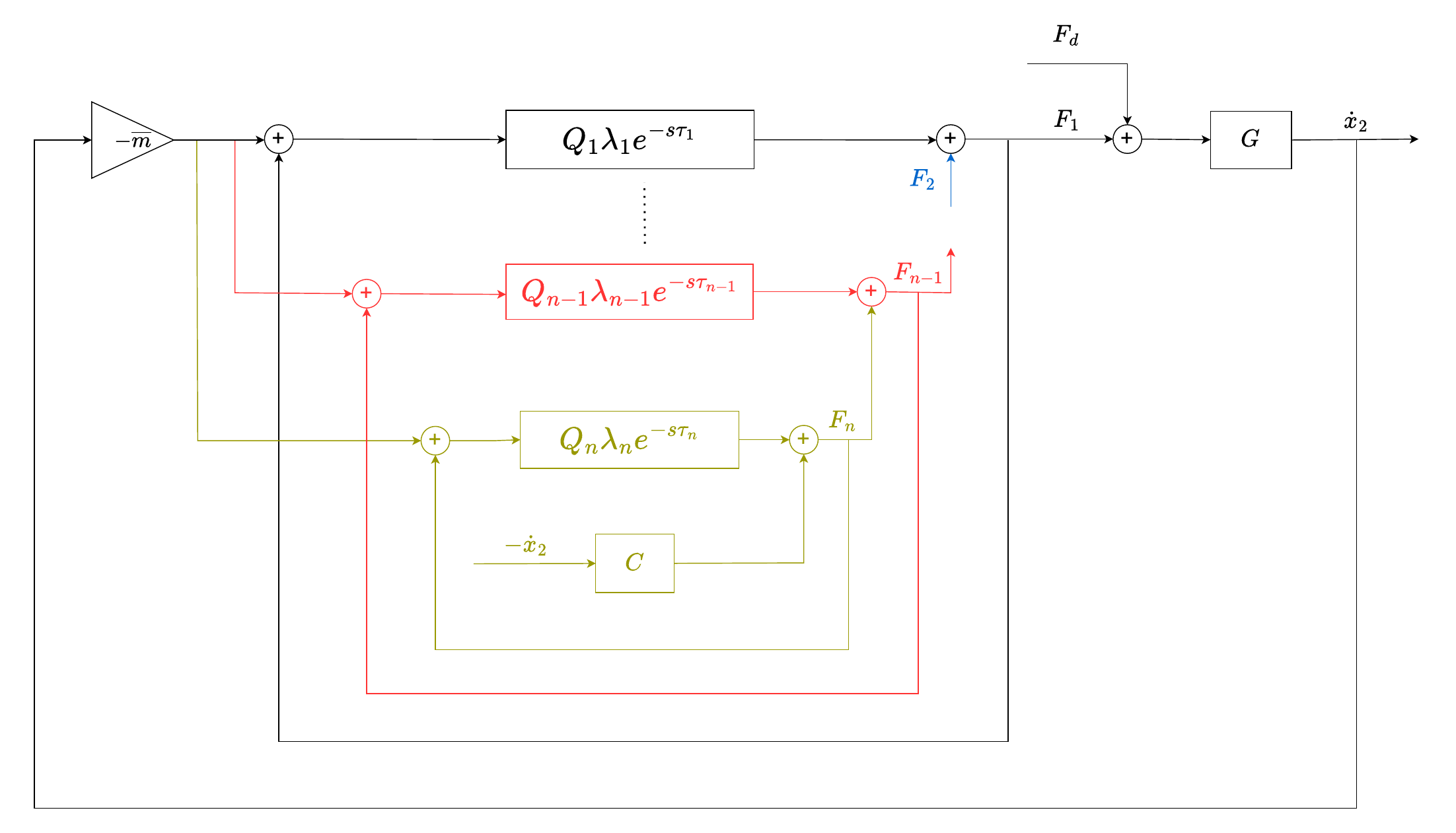}
  \caption{Proposed multi-period repetitive control design}
  \label{fig:2H}
\end{figure}
Rearranging system (\ref{s1}), we can write the total disturbance in terms of the control signal and the controlled signal as $-\phi_1(t)=F_1(t)-\overline m\dot x_2(t)$. Control signal (\ref{s2_time}) can be equivalently interpreted as compensating the total disturbance by its previous $\tau$ value as $F_1(t)=-Q_1\lambda_1\phi_1(t-\tau_1)+F_2(t)$. Substituting $F_1(t)$ in this form,  system (\ref{s1}) can be expressed as $\overline m\dot x_2(t)=\phi_1(t)-Q_1\lambda_1\phi_1(t-\tau_1)+F_2(t)=\phi_2(t)+F_2(t) $, where we  denote $\phi_{i+1}(t)=\phi_i(t)-Q_i\lambda_i\phi_i (t-\tau_i)$  for $i=1,\dots, n-1$. Then we can write  $-\phi_2(t)=F_2(t)-\overline m\dot x_2(t)$ 
and    
 $F_2(t)=Q_2\lambda_2(F_2(t-\tau_2)-\overline m\dot x_2(t-\tau_2))+F_3(t)=
 -Q_2\lambda_2\phi_2(t-\tau_2)+F_3(t) $. Expanding this to $i=n$, and denoting  $F_R$ as the repetitive control and $F_N$ as the nominal control,  the controller (\ref{s2_time}) can be interpreted equivalently in terms of the total disturbance and the nominal control as
\begin{equation}\label{F1_phi}
F_1(t)= \underbrace{\sum_{i=1}^nQ_i\lambda_i (F_i(t-\tau_i)-\overline m\dot x_2(t-\tau_i) )}_{F_R(t)}\underbrace{-C\dot x_2(t)}_{F_N(t)}=
-\sum_{i=1}^nQ_i\lambda_i\phi_i(t-\tau_i)-C\dot x_2(t).
\end{equation}
Note that the controller expression in (\ref{F1_phi}) is a mathematical expression associated with the total disturbance $\phi_i$. It is not implemented based on the total disturbance $\phi_i$; its actual implementation is given in (\ref{s2_time}).
We can write (\ref{F1_phi}) in Laplace form in terms of $\Phi_1(s)$. Take $n=2$ for example, substituting $F_2(t)$ into $F_1(t)$ in terms of the total disturbance we can write $F_1(t)=-Q_1\lambda_1\phi_1(t-\tau_1)-Q_2\lambda_2\phi_2(t-\tau_2)+F_3(t)=  -Q_1\lambda_1\phi_1(t-\tau_1)-Q_2\lambda_2(\phi_1(t-\tau_2)-Q_1\lambda_1\phi_1(t-\tau_1-\tau_2) )+F_3(t)$. In Laplace form, this becomes
$F_1(s)
=(-1+(1-Q_1\lambda_1e^{-s\tau_1})(1-Q_2\lambda_2e^{-s\tau_2}))\Phi_1(s)+F_3(s)
$. For $n=1$, this can be written as $F_1(s)=\big(-1+\prod_{i=1}^n(1-Q_i\lambda_ie^{-s\tau_i})(1-Q_{n+1}\lambda_{n+1}e^{-s\tau_{n+1}})\big)\Phi_1(s)+F_3(s)=(-1+\prod_{i=1}^{n+1}(1-Q_i\lambda_ie^{-s\tau_i}))\Phi_1(s)+F_3(s)
$.  Hence, it follows  for $i=n$
\begin{equation}\begin{split}\label{s_22}
F_1(s)=F_R(s)+F_N(s)=\big(-1+\prod_{i=1}^n(1-Q_i\lambda_ie^{-s\tau_i})\big)\Phi_1(s)-CsX_2(s).
\end{split}\end{equation}
Substituting the  controller expression (\ref{s_22}) into (\ref{s1}) yields the closed-loop system 
\begin{equation}\label{s6_1}
\overline m s X_2(s)=\prod_{i=1}^n(1-Q_i\lambda_ie^{-s\tau_i})\Phi_1(s)-Cs X_2(s).
\end{equation}
Expanding $\Phi_1(s)=-\Delta m s X_2(s)-kX_1(s)-dX_2(s)+F_d(s)$ in (\ref{s6_1}) and then substituting the nominal control $C$, the transfer function from $F_d(s)$ to $s X_2(s)$ becomes
\begin{equation}\label{Fd_x2}
\resizebox{0.98\linewidth}{!}{$
\begin{split}
&sX_2(s)=\frac{\prod_{i=1}^n(1-Q_i\lambda_ie^{-s\tau_i})
}{\overline m+\prod_{i=1}^n(1-Q_i\lambda_ie^{-s\tau_i})(\Delta m+\frac{d}{s}+\frac{k}{s^2})+C
}F_d(s)\\
&=\frac{\prod_{i=1}^n(1-Q_i\lambda_ie^{-s\tau_i})s^2
}{(\overline m+\prod_{i=1}^n(1-Q_i\lambda_ie^{-s\tau_i})\Delta m)s^2+(\prod_{i=1}^n(1-Q_i\lambda_ie^{-s\tau_i})d+g)s+(\prod_{i=1}^n(1-Q_i\lambda_ie^{-s\tau_i})k+p)
}F_d(s)
\end{split}
$}
\end{equation}
We can rewrite (\ref{Fd_x2}) as:
\begin{equation}\label{small gain}
sX_2(s)=\frac{ \frac{\prod_{i=1}^n(1-Q_i\lambda_ie^{-s\tau_i})}{\overline ms^2+gs+p}s^2
}{1+\frac{\prod_{i=1}^n(1-Q_i\lambda_ie^{-s\tau_i})}{\overline ms^2+gs+p}(\Delta ms^2+ds+k)}
F_d(s)
\end{equation}
A sufficient condition by Small Gain Theorem \cite{Zhou} for stability of (\ref{small gain}) is 
\begin{equation}
\left|\frac{\prod_{i=1}^n(1-Q_i\lambda_ie^{-j\omega\tau_i})(\Delta m(j\omega)^2+d(j\omega)+k)}{\overline{m}(j\omega)^2+g(j\omega)+p}\right|_\infty<1
\end{equation}
for all frequencies $\omega$.

Rearranging the total disturbance $\Phi_1(s)=\overline ms X_2(s)-F_1(s)$ and substituting it into (\ref{s_22}),
we derive the transfer function from  $s X_2(s)$ to $F_1(s)$ as
\begin{equation}\label{s7}
\begin{split}
-\frac{F_1(s)}{s X_2(s)}=\frac{\overline m\big(1-\prod_{i=1}^n(1-Q_i\lambda_ie^{-s\tau_i})\big)s^2+gs+p
}{\prod_{i=1}^n(1-Q_i\lambda_ie^{-s\tau_i})s^2}.
\end{split}
\end{equation}
This is the transfer function of the RC controller for $i=n$. For $i=1$, this recovers the transfer function for single RC \cite{Steinbuch}.

\section{Problem formulation}
A prominent feature of the proposed controller  (\ref{s2}) is that both multi-period RC and TDC \cite{Hsia}, \cite{Youcef} can be designed by simply adjusting the control parameter $\tau_i$. 
There are five design parameters in the controller (\ref{s2}), as listed in Table \ref{tab:control} along with their suggested choice.  We explain how to choose those parameters for both RC design and TDC design hereafter.
\begin{table}[h!]
  \begin{center}
  \caption{parameter choices For controller (\ref{s2})}    \label{tab:control}
    \begin{tabular}{|c|c|c|c|} 
\hline   
& \multicolumn{2}{|c|}{RC}&TDC\\ \hline 
\multirow{2}{*}{$Q_i$}&Low-pass filter & Band-pass filter & Low-pass or Band-pass filter \\ \cline{2-4}
& $ \frac{\omega_{c_i}}{s+\omega_{c_i}}$& $\frac{\omega_{c_{1_i}}}{s+\omega_{c_{1_i}}}\frac{s}{s+\omega_{c_{2_i}}}$& Same as defined for RC \\ \hline 
$\lambda_i$& $ \left|Q_i(\omega_{d_i})^{-1}\right|-\frac{\beta_i}{\omega_{c_i}}$& $\left|Q_i(\omega_{d_i})^{-1}\right|-\frac{\beta_i}{|\omega_{c_{1_i}}j\omega_{d_i}|}$ & Same as defined for RC\\ \hline
$\tau_i$&\multicolumn{2}{|c|}{$\frac{2\pi+\textrm{arg}(Q_i(\omega_{d_i}))}{\omega_{d_i}}$}&$t_s$\\
      \hline
$p,\,g$& \multicolumn{3}{|c|}{Nominal control parameters}\\ \hline       
 \end{tabular}
  \end{center}
\end{table}

For RC design,
we consider the fundamental frequency of the disturbance $\omega_{d_i}$ to be  known.
 We consider two options for the filter design of $Q_i$. We can either choose $Q_i(s)=\frac{\omega_{c_i}}{s+\omega_{c_i}}$ as a low-pass filter, where $\omega_{c_i}$ denotes the low-pass cutoff frequency, or choose $Q_i(s)=\frac{\omega_{c_{1_i}}}{s+\omega_{c_{1_i}}}\frac{s}{s+\omega_{c_{2_i}}}$ as a band-pass filter, where $\omega_{c_{2_i}}<\omega_{c_{1_i}}$ and $\omega_{c_{1_i}}$ denotes the high-pass and the low-pass cutoff frequencies, respectively.
We choose the RC gain $\lambda_i =\left|Q_i(\omega_{d_i})^{-1}\right|-\frac{\beta_i}{\omega_{c_i}}$ when $Q_i$ is a low-pass filter, and  $\lambda_i=\left|Q_i(\omega_{d_i})^{-1}\right|-\frac{\beta_i}{|\omega_{c_{1_i}}j\omega_{d_i}|}$ when $Q_i$ is a bandpass filter. Scalar $\beta_i$ can be positive or negative and is a design parameter that determines the level of disturbance rejection.
Its range depends on $\omega_{c_i}$ and $\omega_{d_i}$.
 Its design will be explained in more in-depth detail in Section \ref{beta disturbance rejection}. 
For both filters we choose $\tau_i=\frac{2\pi+\textrm{arg}(Q_i(\omega_{d_i}))}{\omega_{d_i}}$. The nominal control parameters $p$ and $g$ are designed for the RC closed-loop stability.
For TDC design,
it is not required that the disturbance frequencies be known. 
 We choose $\tau_i=t_s$, where $t_s$ is the sampling period of the electromechanical system. The filter $Q_i$ and the gain $\lambda_i$ attain the same formula as RC.

Our objective is as follows:
for multi-period RC, where 
given a periodic disturbance with known frequency $\omega_{d_i}$ and   given a filter $Q_i$, which can  either be a low-pass filter with cutoff frequencies $\omega_{c_i}$ or a band-pass filter with cutoff frequencies $\omega_{c_{1_i}}$, $\omega_{c_{2_i}}$, ($\tau_i$ and $\lambda_i$ can be derived as in Table \ref{tab:control}),
control parameters $\beta_i$ (positive or negative) and the nominal control gain $p,\, g$ 
 will be designed  such that 
 the  trajectories of the closed-loop system (\ref{Fd_x2}) under  control $F_1$ in (\ref{s2_time}) 
 converge to the $M_i(\beta_i)\Delta_i-$neighbourhood of  $\{x_1(\omega_{d_i}),\,x_2(\omega_{d_i}),\, \dot x_2(\omega_{d_i})=0\}$, where $M_i(\beta_i)\geq 0$ is a positive number to be minimized. 
$M_i(\beta_i)=0$ is only possible if $\beta_i=0$ can be chosen, which recovers RC with full gain  $Q_i(\omega_{d_i})\lambda_i=1$. 
For TDC, where
$\tau_i=t_s$, $M_i(\beta_i)=0$ cannot be designed, but it can be minimized for small enough $t
_s$. 
\begin{remark}\label{remark2}
  A cascaded multi-period RC is proposed in \cite{M.Yamada}--\cite{Blanken} as shown in Fig \ref{fig:Cascaded}, where $\hat T$ denotes the identified transfer function of the nominal closed-loop plant $T=\frac{CG}{1+CG}$ from $F_1$ to $\dot x_2$, which is the complementary sensitivity function. 
Learning filters $L_1,\,L_2$ are designed to account for the model mismatch between $T$ and $\hat T$. 
Setting $\hat T=0$ in Fig. \ref{fig:Cascaded} recovers the parallel multi-period RC as proposed in \cite{W.Chang}, \cite{Owens}.
Suppose  all the filters $Q_1,\,Q_2,\, L_1,\,L_2$ are set to 1, the  cascaded RC in the figure, which is outside the dotted box, can be expressed as
\begin{equation}\label{F_cas}
 -\frac{F_1(s)}{s X_2(s)}= \frac{1-\prod_{i=1}^2(1-\lambda_ie^{-s\tau_i})-(1-\hat T)\lambda_1\lambda_2e^{-s(\tau_1+\tau_2)}   }{\prod_{i=1}^2(1-\lambda_ie^{-s\tau_i}) }.
\end{equation}
Note that the additive structure of (\ref{F_cas}) severely differs from the multiplicative structure in (\ref{s7}).
We will investigate the impact of such an additive structure in the following. 

Equation (\ref{F_cas}) can be rearranged   as $F_1(s)-s X_2(s)=-\frac{1-(1-\hat T)\lambda_1\lambda_2e^{-s(\tau_1+\tau_2)}}{\prod_{i=1}^2(1-\lambda_ie^{-s\tau_i})}s X_2(s)   $. Substituting this into 
the dotted box in Fig \ref{fig:Cascaded}, where  $s X_2(s)=G(F_d+C(F_1(s)-s X_2(s)) )$, yields
  \begin{equation}\label{cascaded_01}
  \begin{split}
 s X_2(s)&=\frac{\prod_{i=1}^2(1-\lambda_i e^{-s\tau_i})G}{\prod_{i=1}^2(1-\lambda_ie^{-s\tau_i})+CG\big(1-(1-\hat T)\lambda_1\lambda_2e^{-s(\tau_1+\tau_2)}\big)}F_d(s).\\
\end{split}
 \end{equation} 
 \begin{figure}[H]
\centering
  \includegraphics[width=1\linewidth]{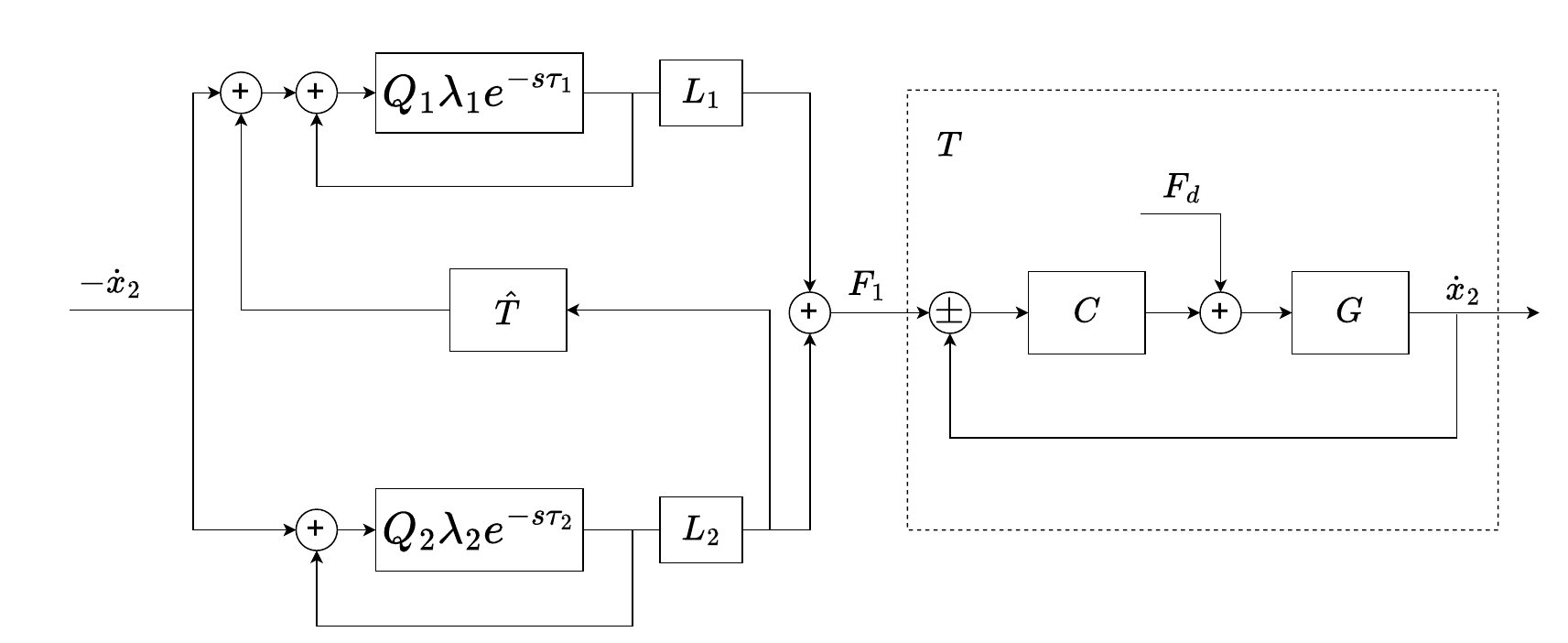}
  \caption{Cascade multi-period RC in  \protect\cite{M.Yamada} -- \protect\cite{Blanken}. $\hat T=0$ represents the parallel multi-period RC in \protect\cite{W.Chang}, \protect\cite{Owens}. }
  \label{fig:Cascaded}
\end{figure}

In cascaded RC, suppose a perfectly identified $\hat T = T$. Then we have $CG(1-T) = T$. 
Denoting $L=CG$ as loop transfer function, $S=(1+L)^{-1}$ as sensitivity function, 
 (\ref{cascaded_01}) can be written as
\begin{equation}\label{cascaded_02}\begin{split}
 s X_2(s)&=\frac{\prod_{i=1}^2(1-\lambda_i e^{-s\tau_i})G}{\prod_{i=1}^2(1-\lambda_ie^{-s\tau_i})+CG-T\lambda_1\lambda_2e^{-s(\tau_1+\tau_2)}}F_d(s)\\
 &= \frac{\prod_{i=1}^2(1-\lambda_i e^{-s\tau_i})G}{\prod_{i=1}^2(1-\lambda_ie^{-s\tau_i})+L(1-S\lambda_1\lambda_2e^{-s(\tau_1+\tau_2)})}F_d(s).
\end{split}
\end{equation}
In cascaded RC for a well identified $\hat T\approx T$, provided that sensitivity function  $S\approx 0$ (or $T\approx 1$) is possible (by that we mean the closed-loop is stable), this will allow to eliminate the interaction term $e^{-s(\tau_1+\tau_2)}$ in (\ref{cascaded_02}). For parallel RC, where $\hat T=0$, the interaction will exist.
Such interaction, however, does not exist in the proposed approach in (\ref{Fd_x2}), where we do not require  $T\approx 1$.

 \begin{figure}[H]
\centering
  \includegraphics[width=1\linewidth]{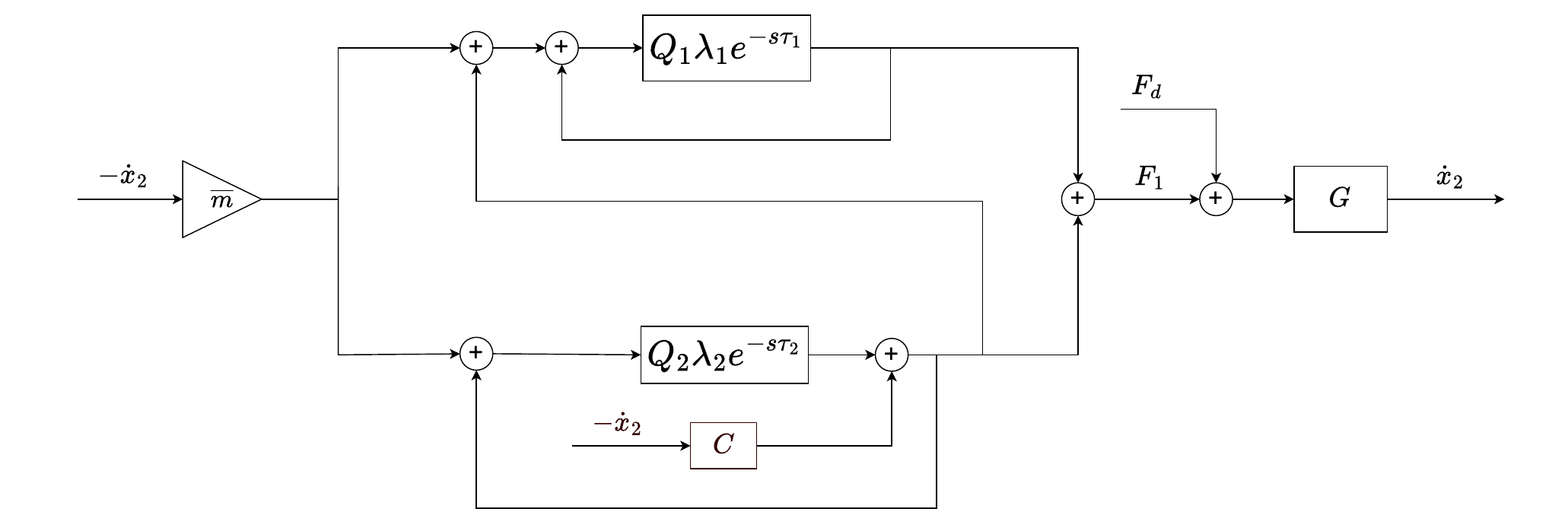}
  \caption{Modified cascade two-period repetitive control without using an identified plant model}
  \label{fig:Cascaded_2RC}
\end{figure}
In fact, we propose a modification of the multi-period RC in Fig. \ref{fig:Cascaded}, as shown in Fig. \ref{fig:Cascaded_2RC}.
   The modification is made by taking $\hat T=1$ (Note this is not the identified closed-loop gain. It is rather just a number here. The transfer function (\ref{Fd_x2}) does not require $T\approx 1$) in Fig. \ref{fig:Cascaded}  and moving the nominal controller $C$ out of $T$ and adding it to the input by the last RC.  
This is equivalent to the proposed architecture in Fig. \ref{fig:2H}, where its closed-loop transfer function is given in (\ref{Fd_x2}).  RC can be turned off by setting $\lambda_i=0$ to recover the nominal control $F_1=-C\dot x_2$, which
pertains to the modified cascaded RC as a plug-in controller. Our approach neither needs an identified plant model nor a high gain control, i.e., $S\approx 0$ (or $T\approx 1$) to eliminate multi-frequency interaction.    This is advantageous over conventional cascaded RC.

\end{remark}

\begin{remark}\label{TDC remark}
The RC and TDC disturbance observer can be shown to have equivalent structures.  
Consider the delay to be the sample period
 $\tau_i=t_s$, which is small for sampling frequency of 500 Hz or 1000 Hz in most of the modern electromechanical or robotic systems. At low frequencies, $st_s\ll 1$, and
using the first-order Taylor approximation, we have
 $e^{-st_s}\approx 1-st_s$. Suppose $|Q_i(\omega_{d_i})\lambda_i|=1$, 
 (\ref{s_22}) can be rewritten as
\begin{equation}\label{TDC_21}\begin{array}{l}
F_1(s)=(-1+(st_s)^n)\Phi_1(s)-Cs X_2(s)
\end{array}
\end{equation}
and (\ref{Fd_x2}) becomes 
\begin{equation}\label{dis compen}
\begin{array}{l}
\big(\overline m+(st_s)^n\Delta m\big)sX_2(s)=-\big(p+(st_s)^nk\big)X_1(s)-\big(g+(st_s)^nd\big)X_2(s)+(st_s)^nF_d(s)
\end{array}
\end{equation}
It is  shown in (\ref{dis compen}) that the disturbance $F_d(s)$ can be compensated for all $s<\frac{1}{t_s}$; larger $n$ lead to larger compensation.  
Controller (\ref{TDC_21}) recovers the TDC disturbance observer \cite{Hsia}, \cite{Youcef}, \cite{b1}.

There are two benefits resulting from including the TDC
controller (\ref{TDC_21})  in the RC. Firstly, RC often requires the exact frequencies of the disturbance to be known. When the disturbance frequency is uncertain, RC performance will deteriorate \cite{Steinbuch}.
TDC can be used in conjunction with RC to compensate for these frequency uncertainties, resulting in suppression at the RC frequencies and enhanced robustness against them. 
Secondly,
since the TDC controller (\ref{TDC_21}) contains the information of the periodic disturbance $F_d(s)$, the reconstructed disturbance signal by TDC helps reconstruct the disturbance frequencies by using frequency identification techniques, 
such as adaptive notch filtering, extended Kalman filter frequency estimation, phase-locked loop technique, adaptive identifier, etc. \cite{Qing} -\cite{Muramatsu}.
Once the disturbance frequency components have been identified, RC can be designed using those frequencies. These two features will be demonstrated in the experimental results.
\end{remark}
\begin{remark}\label{remark1}
A single period high order RC is proposed in \cite{Steinbuch}, \cite{Steinbuch2} as shown in Fig. \ref{fig:0}. It is designed to improve
 the robustness to variation in the disturbance frequency. 
 $\tau$  denotes the period of a single-period disturbance. Parameters $W_i$ must be selected such that $\sum_{i=1}^nW_i=1$ and $\sum_{i=1}^nW_ii^{n-1}=0$. 
Suppose the filter is set to $Q=1$, which leads to the following transfer function
 \begin{equation}\label{steinbuch}
-\frac{F_1(s)}{\overline m sX_2(s)}=\frac{1-(1-e^{-s\tau})^n}{(1-e^{-s\tau})^n}.
 \end{equation}
This is in fact a special case of the more general structure of MPRC proposed in (\ref{s7}) for single period repetitive disturbance, where $|Q_i(\frac{2\pi}{\tau} )\lambda_i|=1$, $\tau_i=\tau$ and $C=0$. Due to the strict constraints on RC gains $W_i$, the same filter $Q$ must be implemented for all RC gains $W_1,\ldots,W_n$ in (\ref{steinbuch}). 
In comparison, different filters can be designed for each RC gain in (\ref{s7}), also seen in Fig. \ref{fig:2H}. 
Besides, RC gain $|Q_i(\omega_{d_i})\lambda_i|$ less than unity can be designed, which is useful when weighting control energy at multiple frequencies. 
This is, however, not possible in Fig. \ref{fig:0} by $W_i$. 
The robustness of RC (\ref{s7}) to the frequency variation of $\omega_{d_i}$ can be significantly enhanced by choosing $\tau_i$ as the period of the disturbance in RC design and $\tau_j=t_s$ in TDC design, since TDC allows suppression at a much wider lower frequency range as discussed in Remark \ref{TDC remark}. We will demonstrate that the robustness of including TDC in RC is more promising than (\ref{steinbuch}) in the experimental results.
\begin{figure}[H]
\centering
  \includegraphics[width=1\linewidth]{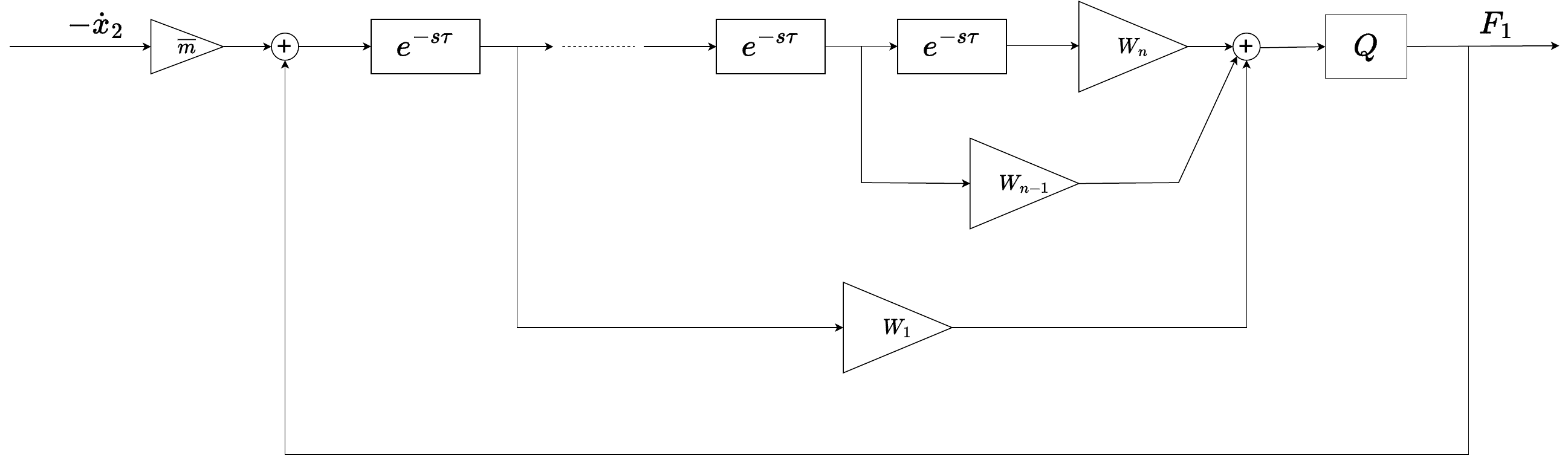}
  \caption{High order repetitive controller design in \protect\cite{Steinbuch} }
  \label{fig:0}
\end{figure}
\end{remark}

\section{\textbf{State-space formulation for $n=2$} }
For the stability analysis of the closed-loop system with RC and nominal control (\ref{s6_1}), we will begin with a state-space formulation using $n=2$. This can be extended to larger $n$ as will be shown later.  Expanding (\ref{s6_1}) we write
\begin{equation}\label{s6_2}
\overline ms X_2(s)=(1-Q_1(s)\lambda_1e^{-s\tau_1})(1-Q_2(s)\lambda_2e^{-s\tau_2})\Phi_1(s)-pX_1(s)-gX_2(s).
\end{equation}
Suppose $Q_1(s)$, $Q_2(s)$ are first order low-pass filters,
 we can write
\begin{equation}\label{lap}
\begin{array}{l}
(s+\omega_{c_1})(s+\omega_{c_2})\big((ms^2+(d+g)s+k+p)X_1(s)-F_d(s)\big)=\big(-\omega_{c_1}\lambda_1(s+\omega_{c_2})e^{-s\tau_1}\\
-\omega_{c_2}\lambda_2(s+\omega_{c_1})e^{-s\tau_2}+\omega_{c_1}\omega_{c_2}\lambda_1\lambda_2e^{-s(\tau_1+\tau_2)}\big)\Phi_1(s).
\end{array}
\end{equation}
Expanding the left-hand side, (\ref{lap}) can be rewritten as
\begin{equation}\label{x21}
\begin{array}{l}
ms^4X_1(s)=-\big(((\omega_{c_1}+\omega_{c_2})m+d+g)s^3+(\omega_{c_1}\omega_{c_2}m+(\omega_{c_1}+\omega_{c_2})(d+g)+k+p)s^2\\+(\omega_{c_1}\omega_{c_2}(d+g)+(\omega_{c_1}+\omega_{c_2})(k+p))s+\omega_{c_1}\omega_{c_2}(k+p)\big)X_1(s)+(s^2+(\omega_{c_1}+\omega_{c_2})s\\+\omega_{c_1}\omega_{c_2})F_d(s)
+\big(-\omega_{c_1}\lambda_1(s+\omega_{c_2})e^{-s\tau_1}
-\omega_{c_2}\lambda_2(s+\omega_{c_1})e^{-s\tau_2}\\+\omega_{c_1}\omega_{c_2}\lambda_1\lambda_2e^{-s(\tau_1+\tau_2)}\big)\Phi_1(s).
\end{array}
\end{equation}
Substituting $\Phi_1(s)=-\Delta msX_2(s)-kX_1(s)-dX_2(s)$,  and denoting $x_3(t)=\dot x_2(t)$, $x_4(t)=\dot x_3(t)$, and vectors $x(t)=\textrm{col}\{x_1(t),\,x_2(t),\,x_3(t),\,x_4(t)\}$ and
  $F_{d_s}(t)=\textrm{col}\{F_d(t)$, $\dot F_d(t)$, $\ddot F_d(t)\}$,
(\ref{x21}) can be put in state-space form as a time-delay system with multiple and interacting delays:
\begin{equation}\label{state_1}
\begin{array}{l}
D_m\dot x=A
x+
A_1x(t-\tau_1)
+A_2
x(t-\tau_2)+
A_3
x(t-\tau_1-\tau_2)\\
+BF_{d_s}(t)+B_1F_{d_s}(t-\tau_1)
+B_2F_{d_s}(t-\tau_2)+B_3F_{d_s}(t-\tau_1-\tau_2),
\end{array}
\end{equation}
where $D_m=\text{diag}\{1,\,1,\,1,\, m\}$ and
\begin{equation}
\resizebox{0.98\linewidth}{!}{$
\begin{array}{l}
A=\begin{pmatrix}
0 &1&0&0\\0&0&1&0\\0&0&0&1\\  A_{41}& A_{42}& A_{43}&A_{44}
\end{pmatrix},\, A_1=\begin{pmatrix}
0 &0&0&0\\0&0&0&0\\0&0&0&0\\A_{1_{41}}&A_{1_{42}}&A_{1_{43}}&A_{1_{44}}\end{pmatrix},\, 
A_2=\begin{pmatrix}
0 &0&0&0\\0&0&0&0\\0&0&0&0\\A_{2_{41}}&A_{2_{42}}&A_{2_{43}}&A_{2_{44}}
\end{pmatrix},\\ A_3=\begin{pmatrix}
0 &0&0&0\\0&0&0&0\\0&0&0&0\\A_{3_{41}}&A_{3_{42}}&A_{3_{43}}&A_{3_{44}}
\end{pmatrix},\,  B=\begin{pmatrix}
0&0&0\\0&0&0\\0&0&0\\B_{41}&B_{42}&B_{43}
\end{pmatrix},\, B_1=\begin{pmatrix}
0&0&0\\0&0&0\\0&0&0\\B_{1_{41}}&B_{1_{42}}&B_{1_{43}}
\end{pmatrix},\\ B_2=\begin{pmatrix}
0&0&0\\0&0&0\\0&0&0\\B_{2_{41}}&B_{2_{42}}&B_{2_{43}}
\end{pmatrix},\, B_3=\begin{pmatrix}
0&0&0\\0&0&0\\0&0&0\\B_{3_{41}}&B_{3_{42}}&B_{3_{43}}
\end{pmatrix},
\end{array}
$}
\end{equation}
with matrix entries given by
\begin{equation}\label{A_matrix}
\begin{array}{l}
\begin{pmatrix}
A_{41}\\A_{42}\\A_{43}\\A_{44}
\end{pmatrix}=-\begin{pmatrix}\underbrace{\begin{pmatrix}
k&0&0\\ d&k&0\\\Delta m&d&k\\0&\Delta m&d
\end{pmatrix}}_{K}+\underbrace{
\begin{pmatrix}
p&0&0\\g&p&0\\\overline m&g&p\\0&\overline m&g
\end{pmatrix}}_{K_p} \end{pmatrix}\underbrace{
\begin{pmatrix}
\omega_{c_1}\omega_{c_2}\\ \omega_{c_1}+\omega_{c_2}\\1
\end{pmatrix}}_{\overline \omega}\\

\begin{pmatrix}
A_{1_{41}}\\A_{1_{42}}\\A_{1_{43}}\\ A_{1_{44}}
\end{pmatrix}=
\omega_{c_1}\lambda_1
K(:,1:2)\underbrace{
\begin{pmatrix}
\omega_{c_2}\\1
\end{pmatrix}}_{\overline \omega_2},\,\,\,\,\, 
\begin{pmatrix}
A_{2_{41}}\\A_{2_{42}}\\A_{2_{43}}\\ A_{2_{44}}
\end{pmatrix}=\omega_{c_2}\lambda_2
K(:,1:2)\underbrace{
\begin{pmatrix}
\omega_{c_1}\\1
\end{pmatrix}}_{\overline \omega_1},\\
\begin{pmatrix}
A_{3_{41}}\\A_{3_{42}}\\A_{3_{43}}\\A_{3_{44}}
\end{pmatrix}=-\omega_{c_1}\omega_{c_2}\lambda_1\lambda_2
K(:,1),\,\,\, \begin{pmatrix}
B_{41}\\B_{42}\\B_{43}
\end{pmatrix}=\bar \omega,\,\,\, \begin{pmatrix}
B_{1_{41}}\\B_{1_{42}}\\B_{1_{43}}
\end{pmatrix}=-\omega_{c_1}\lambda_1I_3(:,1:2)\bar \omega_2,\\
\begin{pmatrix}
B_{2_{41}}\\B_{2_{42}}\\B_{2_{43}}
\end{pmatrix}=-\omega_{c_2}\lambda_2I_3(:,1:2)\bar \omega_1,\,\,\, \begin{pmatrix}
B_{3_{41}}\\B_{3_{42}}\\B_{3_{43}}
\end{pmatrix}=\omega_{c_1}\omega_{c_2}\lambda_1\lambda_2I_3(:,1).
\end{array}
\end{equation}
 The matrices in  (\ref{A_matrix}) can be   expanded for $n>2$ by the following rule:
\begin{remark}
The first column of $K$ always start with $\begin{pmatrix}
k&d&\Delta m
\end{pmatrix}^T$. The second column is a repetition of the first column, but with a zero in the first row. This repetition propagates to the third column. The number of rows and columns in $K$ is $n+2$ and $n+1$, respectively. 
 Matrix $K_p$ follows the same pattern but with $\begin{pmatrix}p&g& \overline m\end{pmatrix}^T$. 
The terms in $\overline \omega$ are the convolution of the $1+\overline \omega_{c_i}$, resulting from the multiplication of polynomials.
In $\omega_{c_1}\lambda_1K(:,1:2)\overline \omega_2$, where $K(:,1:2)$ denotes the matrix with the entire first and second columns of matrix $K$,
the terms $\omega_{c_1}\lambda_1$ and  $\overline \omega_2$
 relate to $\omega_{c_1}\lambda_1(s+\omega_{c_2})e^{-s\tau_1}$ in (\ref{x21}). The parameters in $\overline \omega_2$ correspond to the parameters of the polynomial $s+\omega_{c_2}$ in reverse order.
$K(:,1:2)$ keeps the first two columns in $K$  as its number of columns needs to be equal to the number of rows in $ \omega_2$.
  The same rules apply to the parameters of $A_2$ and $A_3$.
Following these rules, an extension to case $n=3$ can be derived in Appendix \ref{Appen_A}.
\end{remark}

Denoting 
$\bar A=A+A_1+A_2+A_3$, $\bar B=B+B_1+B_2+B_3$, $B_s=\begin{pmatrix}
\bar B& -B_1&-B_2&-B_3
\end{pmatrix}$ 
and $\overline F_{d}(t)=\textrm{col}\{F_{d_s}(t),\, \int_{t-\tau_1}^t\dot F_{d_s}(\mu)\textrm{d}\mu$, $\int_{t-\tau_2}^t\dot F_{d_s}(\mu)\textrm{d}\mu$, $\int_{t-\tau_1-\tau_2}^t\dot F_{d_s}(\mu)\textrm{d}\mu\}$, (\ref{state_1}) can be rewritten as a distributed time delay system I converted to this distributed delay system by myself. This is a simple conversion. For example, $x(t-\tau)=x(t)-\int_{t-\tau}^t\dot x(\mu)\text{d}\mu$. 
\begin{equation}\label{state space}
\begin{array}{l}
D_m\dot x(t)=\bar Ax(t)-A_1\int_{t-\tau_1}^t\dot x(\mu)\textrm{d}\mu
-A_2\int_{t-\tau_2}^t\dot x(\mu)\textrm{d}\mu-A_3\int_{t-\tau_1-\tau_2}^t\dot x(\mu)\textrm{d}\mu+B_s\overline F_{d}(t).
\end{array}
\end{equation}
From (\ref{state space}), we can see the system stability is determined by system matrices $D_m$, $\bar A$, $A_1$, $A_2$, $A_3$. The  disturbance $B_s(\textrm{end},:)\overline F_d(t)$, where $B_s(\textrm{end},:)$ denotes the last row of matrix $B_s$, is the result of multiple RC control interactions, where it is suppressed at the designated frequencies (e.g. $1-e^{-s\tau_i}=0$ for $\omega_i=n\frac{2\pi}{\tau_i}$ rad/s, $n=1,\,2,\,3,\ldots$) in (\ref{Fd_x2})
but is amplified at the intermediate frequencies (e.g. $1-e^{-s\tau_i}=2$ for $\omega_i=n\frac{\pi}{\tau_i}$ rad/s, $n=1,\,3,\,5,\ldots$), where we have assumed $Q_i\lambda_i=1$. 
Depending on the filter cut-off frequencies and the disturbance frequencies to be suppressed, the RC gain $\lambda_i$ design should account for this to ensure any necessary compensation for the filters' roll-off. 
Therefore, this disturbance is bounded by the filter cut-off frequencies $\omega_{c_1}$, $\omega_{c_2}$ and the RC gains $\lambda_1$, $\lambda_2$. By the formula of the RC gains in Table \ref{tab:control}, they depend on the design of $\beta_i$.  
In the following, we provide conditions for designing $\beta_1$, $\beta_2$ which minimize the disturbance $B_s(\textrm{end},:)\overline F_d(t)$ 
 based on a given  frequency $\omega_{d_x}$  and given filter cut-off frequencies. The frequency $\omega_{d_x}$ could be the disturbance frequencies $\omega_d$, where larger suppression is desirable, or any other intermediate frequencies,  should they be optimized in addition to $\omega_d$.
 We denote $B_s(\textrm{end},:)\overline F_d(t)$ as RC-induced-disturbance since it is the disturbance induced by the RC control. 
\begin{remark}\label{bandState}
In light of the state-space formulation for $n=2$ with two low-pass filters,
we can formulate a state-space system using a band-pass filter for $n=1$ as follows.
Suppose a band-pass filter is designed as
$
Q_1(s)=\frac{\omega_{c_{1_1}}}{s+\omega_{c_{1_1}}}\frac{s}{s+\omega_{c_{2_1}}}
$, where $\omega_{c_{2_i}}<\omega_{c_{1_i}}$ and $\omega_{c_{1_i}}$ denotes the high-pass cutoff frequency and low-pass cutoff frequency, respectively.
Then for single RC, (\ref{s6_2}) becomes
\begin{equation}\label{q_b}
\overline m\dot x_2=(1-Q_{1}(s)\lambda_1 e^{-s\tau})\Phi_1(s)-px_1-gx_2,
\end{equation} 
and (\ref{lap}) can be rewritten as
\begin{equation}
(s+\omega_{c_{1_1}})(s+\omega_{c_{2_1}})\big((ms^2+(d+g)s+k+p)X_1(s)-F_d(s)\big)=-\lambda_1 \omega_{c_1} e^{-s\tau}s\Phi_1(s).
\end{equation}
For state space formulation with a single band-pass filter,   (\ref{state space}) can be written with $A_1,\,A_2=0$. We denote $\tau=\tau_1+\tau_2$ and replace the parameters of $A_3$ by the parameters in $A_1$ in (\ref{A_matrix}) where  $\bar \omega_2=\begin{pmatrix}
0\\1\end{pmatrix}$, and $\omega_{c_1}$ by $\omega_{c_{1_1}}$, $\omega_{c_2}$ by $\omega_{c_{2_1}}$, then (\ref{state space}) can be considered for single RC with a band-pass filter.
\end{remark}
 \section{\textbf{Design of $\beta_i$ for RC-induced-disturbance}}\label{beta disturbance rejection}
  This section provides conditions for designing $\beta_i$ to achieve the maximum possible suppression at a given frequency, given the disturbance frequencies and the filters' cut-off frequencies. Both the low-pass and band-pass filter cases are considered. 
For a filter $Q_i$ with given cut-off frequencies, we consider how to design $\lambda_i$ such that the product of the two $Q_i\lambda_i$ yields the maximum suppression for an RC-induced disturbance at a given frequency.
To do that, we find parameter $\beta_i$  in $\lambda_i$, which compensates for the filter gain $Q_i(s)$,  such that the $B_s(\textrm{end},:)\overline F_d(s)$ in (\ref{state space})  at a given frequency $\omega_{d_x}$ is minimized. Note $\omega_{d_x}$ could  coincide with the disturbance frequency $\omega_{d_x}=\omega_{d_i}=\frac{2\pi}{\tau_i}$ rad/s or
be any other frequencies $\omega_{d_x}\neq\omega_{d_i}$ where optimization is of interest.  We denote $F_{d_R}=B_s(\textrm{end},:)\overline F_d(s)$ as  the RC-induced-disturbance,  where the non-zeros terms in $B_s$ are considered.
 \subsection{low-pass filter}
Suppose all filters in the RC are low-pass filters whose formula is defined in Table \ref{tab:control},
expanding $B_s(\textrm{end},:)\overline F_d(s)$ according to  $B_s,\,\overline F_d$ and $F_{d_s}$, we can write the RC-induced-disturbance in the following  Laplace form  for the case of $i=n$
     
 \begin{equation}\label{remark7_1}
F_{d_R}=B_s(\textrm{end},:)\overline F_d(s)=\prod_{i=1}^n(s+\omega_{c_i}(1-\lambda_ie^{-s\tau_i}))F_d(s).
 \end{equation}
 This shows that the RC-induced disturbance by the proposed multi-period RC is the product of the RC-induced disturbance induced by each RC loop. 
  Let us denote a positive number $\gamma_i(\beta_i)\geq 0$.  For a specific frequency, the following holds
  \begin{equation}\label{rem_1}
s+\omega_{c_i}(1-\lambda_ie^{-s\tau_i})=\gamma_i (s+\omega_{c_i}),\hspace{.1in} \textrm{i.e.}\hspace{.1in}\gamma_i=1-Q_i(s)\lambda_ie^{-s\tau_i}.
\end{equation}
The intuition to design $\gamma_i$ in this way can be understood by substituting it into the closed-loop system (\ref{s6_1}) so that we can write $\overline msX_2(s)=\prod_{i=1}^n\gamma_i\Phi_1(s)-CsX_2(s)$. Our objective is to minimize the influence of the total disturbance $\Phi_1(s)$ on system outputs by minimizing $\gamma_i$ for a specific frequency.

Therefore, it is desirable to minimize $\gamma_i$ in order to minimize the RC-induced disturbance. 
   (\ref{remark7_1}) can be rewritten as
\begin{equation}\label{distur_form}
F_{d_R}=\prod_{i=1}^{n}\big(\gamma_i(s+\omega_{c_i})\big)F_d(s).
\end{equation}
 Recall from Table \ref{tab:control} that $\lambda_i=\left|\frac{j\omega_{d_i}+\omega_{c_i}}{\omega_{c_i}}\right|-\frac{\beta_i}{\omega_{c_i}}$ and $\tau_i=\frac{2\pi+\textrm{arg}(Q_i(\omega_{d_i}))}{\omega_{d_i}}$.
Let us denote $\omega_{d_x}$ as a designed frequency by the user, and denote the argument of the following transfer function as 
\begin{equation}\label{rem_1_1}
\textrm{arg}\left(\frac{\omega_{c_i}\lambda_i}{j\omega_{d_x}+\omega_{c_i}}  e^{-j\omega_{d_x}\tau_i}\right)=-\textrm{tan}^{-1}\left(\frac{\omega_{d_x}}{\omega_{c_i}}\right)
-\frac{\omega_{d_x}}{\omega_{d_i}}\left(2\pi-\textrm{tan}^{-1}\left(\frac{\omega_{d_i}}{\omega_{c_i}}\right)\right)\triangleq \alpha_{i}.
\end{equation} 
Then (\ref{rem_1}) can be rewritten as
\begin{equation}\label{gamma_j}
\gamma_i=1-\left|\frac{\omega_{c_i}\lambda_i}{j\omega_{d_x}+\omega_{c_i}}\right|e^{j\alpha_{i}}=1-\frac{|j\omega_{d_i}+\omega_{c_i}|-\beta_i}{|j\omega_{d_x}+\omega_{c_i}|}e^{j\alpha_{i}}.
\end{equation}
Therefore, we can write
\begin{equation}\label{gamma_jabs}
\begin{split}
&\gamma_i^2=\frac{\big(|j\omega_{d_x}+\omega_{c_i}|-(|j\omega_{d_i}+\omega_{c_i}|-\beta_i)\textrm{cos}(\alpha_i)\big)^2+\big((|j\omega_{d_i}+\omega_{c_i}|-\beta_i)\textrm{sin}(\alpha_i)\big)^2}{|j\omega_{d_x}+\omega_{c_i}|^2}.
\end{split}
\end{equation}
 Denoting $a_{i}=|j\omega_{d_x}+\omega_{c_i}|-|j\omega_{d_i}+\omega_{c_i}|b_{i}$, $b_{i}=\textrm{cos}(\alpha_i)$, $c_{i}=\frac{1}{|j\omega_{d_x}+\omega_{c_i}|}$, $e_{i}=-|j\omega_{d_i}+\omega_{c_i}|f_{i} $, $f_{i}=-\textrm{sin}(
 \alpha_i)$, which are all known quantities,  for a positive number $r_i>0$,  we can write
 \begin{equation}\label{gamma_js}
 \gamma_i^2=c_{i}^2(a_{i}+b_{i}\beta_i)^2+c_{i}^2(e_{i}+f_{i}\beta_i)^2\leq r_i^2,\hspace{.1in} \lambda_i\geq 0,
 \end{equation}
 if
 \begin{equation}\label{LMI_gamma}
\begin{pmatrix}
-r_i^2+c_{i}^2(a_{i}^2+e_{i}^2)+2c_{i}^2(a_{i}b_{i}+e_{i}f_{i})\beta_i&\sqrt{c_{i}^2(b_{i}^2+f_{i}^2)}\beta_i\\ \ast& -1
\end{pmatrix}<0,\,\ \ \textrm{and}\,\ \ 
\frac{e_{i}}{f_{i}}+\beta_i<0.
 \end{equation}
 For $\omega_{d_x}=\omega_{d_i}$, we have $\alpha_i=-2\pi$. This implies $\gamma_i=c_i\beta_i=0$ for $\beta_i=0$ in (\ref{gamma_j}). This means there is zero RC-induced-disturbance at $\omega_{d_i}$ with $\beta_i=0$. This is expected, as RC has an infinite control gain at the repetitive frequencies. 
For minimizing repetitive disturbance at frequency    $\omega_{d_x}\neq \omega_{d_i}$ using the $i$-th RC where $\tau_i=\frac{2\pi}{\omega_{d_i}}$,
   it is necessary to compute $\beta_i\neq 0$ via minimization of $r_i$ in LMIs (\ref{LMI_gamma}), resulting in non-zero RC-induced-disturbance at $\omega_{d_i}$. For TDC with $\tau_i=t_s$, this can also be considered following the same procedure.
In which case, the RC-induced disturbance can never be zero as it is proportional to $t_s$.

\subsection{band-pass filter}\label{sec_bandpa}
Recall the band-pass filter  $Q_i(s)=\frac{\omega_{c_{1_i}}s}{(s+\omega_{c_{1_i}})(s+\omega_{c_{2_i}})}$, where $\omega_{c_{2_i}}<\omega_{c_{1_i}}$ as given in Table \ref{tab:control}, the RC-induced-disturbance  can be written as
\begin{equation}\label{band_1}
F_{d_R}=B_s(\textrm{end},:)\overline F_d(s)=\prod_{i=1}^n \big((s+\omega_{c_{1_i}})(s+\omega_{c_{2_i}})-\omega_{c_{1_i}}s\lambda_ie^{-s\tau_i}
\big)F_d(s).
\end{equation}For a specific frequency, the following holds
\begin{equation}\label{band_11}
(s+\omega_{c_{1_i}})(s+\omega_{c_{2_i}})-\omega_{c_{1_i}}s\lambda_ie^{-s\tau_i}=\gamma_i (s+\omega_{c_{1_i}})(s+\omega_{c_{2_i}}),\,\, \textrm{i.e.}\,\, \gamma_i=1-Q_i(s)\lambda_ie^{-s\tau_i},
\end{equation}
we can write (\ref{band_1}) as
\begin{equation}\label{distur_form_BPF}
F_{d_R}=\prod_{i=1}^n\big(  \gamma_i(s+\omega_{c_{1_i}})(s+\omega_{c_{2_i}})       \big)F_d(s).
\end{equation}
Let $\lambda_i=\left|\frac{j\omega_{d_i}+\omega_{c_{1_i}}}{\omega_{c_{1_i}}}\frac{j\omega_{d_i}+\omega_{c_{2_i}}}{j\omega_{d_i}}\right|-\frac{\beta_i}{|\omega_{c_{1_i}}j\omega_{d_i}|}$ and
 $\tau_i=\frac{2\pi+\textrm{arg}(Q_{i}(\omega_{d_i}))}{\omega_{d_i}}$, as given by Table \ref{tab:control}.
For a given frequency $\omega_{d_x}$, let's denote  
\begin{equation}
\resizebox{0.98\linewidth}{!}{$
\begin{split}
\textrm{arg}\left(  \frac{\omega_{c_{1_i}}j\omega_{d_x}}{(j\omega_{d_x}+\omega_{c_{1_i}})(j\omega_{d_x}+\omega_{c_{2_i}})}\lambda_ie^{-j\omega_{d_x}\tau_i}  \right)&=\frac{\pi}{2}-\textrm{tan}^{-1}\left(\frac{\omega_{d_x}}{\omega_{c_{1_i}}}\right)-\textrm{tan}^{-1}\left(\frac{\omega_{d_x}}{\omega_{c_{2_i}}}\right)\\&-\frac{\omega_{d_x}}{\omega_{d_i}}\left(\frac{5\pi}{2}-\textrm{tan}^{-1}\left(\frac{\omega_{d_i}}{\omega_{c_{1_i}}}\right)-\textrm{tan}^{-1}\left(\frac{\omega_{d_i}}{\omega_{c_{2_i}}}\right)\right)\triangleq\alpha_i.                             
\end{split}
$}
\end{equation}
Then  (\ref{band_11}) can be written equivalently as
\begin{equation}
\begin{split}
\gamma_i&=1-\left|\frac{\omega_{c_{1_i}}j\omega_{d_x}}{(j\omega_{d_x}+\omega_{c_{1_i}})(j\omega_{d_x}+\omega_{c_{2_i}})}\right|\lambda_ie^{j\alpha_i} \\
&=1-\frac{|j\omega_{d_x}|  \left(\left|(j\omega_{d_i}+\omega_{c_{1_i}})(j\omega_{d_i}+\omega_{c_{2_i}})\right|-\beta_i\right)}{\left|(j\omega_{d_x}+\omega_{c_{1_i}})(j\omega_{d_x}+\omega_{c_{2_i}})j\omega_{d_i}\right| }e^{j\alpha_i}.
\end{split}
\end{equation}
Denoting \scalebox{0.85}{$a_{i}=\left|(j\omega_{d_x}+\omega_{c_{1_i}}) (j\omega_{d_x}+\omega_{c_{2_i}}) j\omega_{d_i}\right|- \left|(j\omega_{d_i}+\omega_{c_{1_i}}) (j\omega_{d_i}+\omega_{c_{2_i}})\right|b_{i}$}, $b_{i}=|j\omega_{d_x}|\textrm{cos}(\alpha_i)$, $c_{i}=\frac{1}{\left|(j\omega_{d_x}+\omega_{c_{1_i}}) (j\omega_{d_x}+\omega_{c_{2_i}}) j\omega_{d_i}\right|} $,
$e_{i}=-\left|(j\omega_{d_i}+\omega_{c_{1_i}}) (j\omega_{d_i}+\omega_{c_{2_i}}) \right|f_{i}$, $f_{i}=-|j\omega_{d_x}|\textrm{sin}(\alpha_i)$,
  $\beta_i$ can be computed  in LMIs (\ref{LMI_gamma})  that minimizes $r_i$. 

\begin{proposition}\label{proposition1}
Given a disturbance frequency $\omega_{d_i}$,  low-pass cutoff frequency $\omega_{c_i}$ for a low-pass filter (or low-pass and high-pass cutoff frequencies $\omega_{c_{1_i}}$, $\omega_{c_{2_i}}$ for a band-pass filter),  
for any frequency $\omega_{d_x}$,
scalar $\beta_i$ can be computed such that LMIs in 
(\ref{LMI_gamma}) are feasible for large enough $r_i>0$. Then the RC-induced-disturbance  $F_{d_R}(\omega_{d_x})$ with $\lambda_i=\left|Q_i(\omega_{d_i})^{-1}\right|-\frac{\beta_i}{\omega_{c_i}}$ if $Q_i$ is a low-pass filter (or $\lambda_i=\left|Q_i(\omega_{d_i})^{-1}\right|-\frac{\beta_i}{|\omega_{c_{1_i}}j\omega_{d_i}|}$ if $Q_i$ is a band-pass filter), 
and $\tau_i=\frac{2\pi+\textrm{arg}(Q_i(\omega_{d_i}))}{\omega_{d_i}}$ for RC (or $\tau_i=t_s$ for TDC),  
is minimized for small enough $r_i$. Moreover, for the case of RC, where $\omega_{d_x}=\omega_{d_i}$, the 	LMIs are feasible for arbitrarily small $r_i>0$ which yields  $\beta=0$. 
Denoting

\begin{subequations}\label{eq:gamma_bar}
\begin{empheq}[left={ \overline{\gamma} = \empheqlbrace }]{align}
&\prod_{i=1}^{n} r_i |j\omega_{d_x} + \omega_{c_i}|, && \text{if } Q_i \text{ is a low-pass filter,} \label{eq:gamma_bar_a} \\
&\prod_{i=1}^{n} r_i |j\omega_{d_x} + \omega_{c_{1_i}}| \, |j\omega_{d_x} + \omega_{c_{2_i}}|, && \text{if } Q_i \text{ is a band-pass filter,} \label{eq:gamma_bar_b}
\end{empheq}
\end{subequations}
the RC-induced disturbance is then bounded by 
$|F_{d_R}(\omega_{d_x})|\leq \overline \gamma |F_d(\omega_{d_x})|\leq\overline \gamma\Delta_x$, where $|F_d(\omega_{d_x})|\leq \Delta_x$.
\end{proposition}

\section{\textbf{Input to State Stability of RC system}}\label{stability}
Next, we will consider the stability of the closed-loop system (\ref{state space}) under both RC and nominal control. 
Given Proposition \ref{proposition1} and denoting $ B_d= \begin{pmatrix}0&0&0&  1\end{pmatrix}^T$,
we define the vector norm $\|B_s\overline F_d(s)\|\leq \|B_d\overline \gamma F_d(s)\|$.
We consider the following neutral distributed delay system, where $B_s\overline F_d$ is replaced by  $B_d\overline \gamma F_d$  in (\ref{state space}).
\begin{equation}\label{st1}
\resizebox{0.98\linewidth}{!}{$
\frac{\textrm{d}}{\textrm{d}t}(x+D_m^{-1}A_1\int_{t-\tau_1}^tx(\mu)\textrm{d}\mu+D_m^{-1}A_2\int_{t-\tau_2}^tx(\mu)\textrm{d}\mu+D_m^{-1}A_3\int_{t-\tau_1-\tau_2}^tx(\mu)\textrm{d}\mu)=D_m^{-1}(\bar Ax+B_d\overline \gamma F_d)
$}
\end{equation}The purpose of this reformulation is to apply the Integral Inequality in Lemma 3
 The following lemma will be useful for the analysis of exponential stability and Input-to-State Stability (ISS).
\begin{lemma}\label{input-to-state}
\cite{Fridman1} Let $V:\,[0,\infty)\rightarrow \mathbb{R}^+$ be an absolutely continuous function. If there exist $\alpha>0$ and $b>0$ such that the derivative of $V$ satisfies the following inequality
\begin{equation}
\frac{\textrm{d}}{\textrm{d}t}V+2\alpha V-b\|w\|^2\leq 0
\end{equation} 
then it follows that for all $\|w\|\leq \Delta$
\begin{equation}
V(t)\leq e^{-2\alpha (t-t_0)}V(t_0)+(1-e^{-2\alpha (t-t_0)})\frac{b}{2\alpha}\Delta^2
\end{equation}
for $t\rightarrow \infty$.
\end{lemma}
We introduce the following lemmas as a basis for our stability proof.
\begin{lemma}\label{Extended Jensen}
 For any positive definite matrix $R\in \mathbb{R}^{m\times m}$, a scalar $\tau>0$ and a vector function $w\,:\,[-\tau,0]\rightarrow\mathbb{R}^m$ such that the integrations concerned are well defined, the following holds:
\begin{equation}\label{Jensen}
\int_{-\tau}^0w^T(\mu)Rw(\mu)\textrm{d}\mu\geq \frac{1}{\tau}\int_{-\tau}^0w^T(\mu)\textrm{d}\mu R\int_{-\tau}^0w (\mu)\textrm{d}\mu,
\end{equation}
\begin{equation}\label{Extend Jensen}
\int_{-\tau}^0\int_{t+\theta}^tw^T(\mu)Rw(\mu)\textrm{d}\mu\textrm{d}\theta\geq \frac{2}{\tau^2}\left(\int_{-\tau}^0\int_{t+\theta}^tw^T(\mu)\textrm{d}\mu\textrm{d}\theta\right)R\left(\int_{-\tau}^0\int_{t+\theta}^tw(\mu)\textrm{d}\mu\textrm{d}\theta\right).
\end{equation}
\end{lemma}
Inequality (\ref{Jensen}) corresponds to Jensen's inequality, and  (\ref{Extend Jensen}) corresponds to Extended Jensen's inequality  \cite{Sun09}.
\begin{lemma}\label{Wintinger}
(Wirtinger inequality \cite{Seuret13}) For a given  positive definite matrix $R\in \mathbb{R}^{m\times m} $ and any continuously differentiable function $w\,:\, [a,\,b]\rightarrow \mathbb{R}^m$ ($w$ is redefined from Lemma \ref{Extended Jensen}), the following inequality holds:
\begin{equation}
\int_a^b\dot w (\mu)R\dot w(\mu)\textrm{d}\mu\geq \frac{1}{b-a}\Omega_0^TR\Omega_0+\frac{3}{b-a}\Omega_1^TR\Omega_1,
\end{equation}
where $\Omega_0=w(b)-w(a)$, $\Omega_1=w(b)+w(a)-\frac{2}{b-a}\int_a^bw(\mu)\textrm{d}\mu$.
\end{lemma}
Firstly, let us define two column vectors:
\begin{equation}
\begin{array}{l}
\xi=\textrm{col}\{x,\, x(t-\tau_1),\, x(t-\tau_2),\, x(t-\tau_1-\tau_2),\\ \displaystyle\frac{1}{\tau_1}\displaystyle\int_{t-\tau_1}^tx(\mu)\textrm{d}\mu,\, \frac{1}{\tau_2}\displaystyle\int_{t-\tau_2}^tx(\mu)\textrm{d}\mu,\, \frac{1}{\tau_1+\tau_2}\displaystyle\int_{t-\tau_1-\tau_2}^tx(\mu)\textrm{d}\mu,\,
  \overline \gamma F_d
\},\\\\
\xi_1=\textrm{col}\{x,\, \displaystyle\int_{t-\tau_1}^tx(\mu)\textrm{d}\mu,\, \displaystyle\int_{t-\tau_2}^tx(\mu)\textrm{d}\mu,\, \displaystyle\int_{t-\tau_1-\tau_2}^tx(\mu)\textrm{d}\mu\}.
\end{array}
\end{equation}
In the following, we denote $G^{m,8}\{{\scriptstyle{x,y}}=I_4\}$ or $G^{m,4}\{{\scriptstyle{x,y}}=I_4\}$ as matrices, where superscript $m$ is such that $m\times 4$ are the row numbers of the two matrices where $4$ corresponds to the row number of vector $x$.
Superscript 8 and 4 correspond to the number of variables in $\xi$ and $\xi_1$, respectively, and $8\times 4$, $4\times 4$ are the number of columns of the two matrices. 
${\scriptstyle{x,y}}=I_4$ means the matrix corresponding to the $x$-th  index in $m$ and the $y$-th variable index in the superscript $8$ or $4$  of matrix $G$ is an identity matrix.
All other terms in the matrices are zero matrices with appropriate dimensions.

 Let positive definite matrices $R_1,\,R_2,\,R_3\in\mathbb{R}^{4\times 4}>0$ and $S_1,\,S_2,\,S_3\in\mathbb{R}^{4\times 4}>0$, 
we define the following matrices:
\begin{equation}
\resizebox{0.98\linewidth}{!}{$
\begin{split}
G_0^{4,8}&=\{{\scriptstyle{1,1}}=I_4,\, {\scriptstyle{1,5}}=\tau_1D_m^{-1}A_1,\, {\scriptstyle{1,6}}=\tau_2D_m^{-1}A_2,\, {\scriptstyle{1,7}}=(\tau_1+\tau_2)D_m^{-1}A_3,\, {\scriptstyle{2,5}}=\tau_1I_4,\\ & \hspace{.3in}{\scriptstyle{3,6}}=\tau_2I_4,\, {\scriptstyle{4,7}}=(\tau_1+\tau_2)I_4
\};\\
G_1^{4,8}&=\{{\scriptstyle{1,1}}=D_m^{-1}\bar A,\, {\scriptstyle{1,8}}=D_m^{-1}B_d,\,
{\scriptstyle{2,1}}=I_4,\,{\scriptstyle{2,2}}=-I_4,\,{\scriptstyle{3,1}}=I_4,\, {\scriptstyle{3,3}}=-I_4,\, {\scriptstyle{4,1}}=I_4,\,{\scriptstyle{4,4}}=-I_4\};\\
G_2^{1,8}&=\{ {\scriptstyle{1,1}}=D_m^{-1}A,\,{\scriptstyle{1,2}}=D_m^{-1}A_1,\,{\scriptstyle{1,3}}=D_m^{-1}A_2,\,{\scriptstyle{1,4}}=D_m^{-1}A_3,\, {\scriptstyle{1,8}}=D_m^{-1}B_d
\};\\
G_3^{6,8}&=\{ {\scriptstyle{1,1}}=I_4,\,{\scriptstyle{1,2}}=-I_4,\,{\scriptstyle{2,1}}=I_4,\,{\scriptstyle{2,2}}=I_4,\,{\scriptstyle{2,5}}=-2I_4,\,{\scriptstyle{3,1}}=I_4,\,{\scriptstyle{3,3}}=-I_4,\, {\scriptstyle{4,1}}=I_4,\, {\scriptstyle{4,3}}=I_4,\\ &\hspace{.3in}{\scriptstyle{4,6}}=-2I_4,\, {\scriptstyle{5,1}}=I_4,\,{\scriptstyle{5,4}}=-I_4,\,{\scriptstyle{6,1}}=I_4,\,{\scriptstyle{6,4}}=I_4,\, {\scriptstyle{6,6}}=-2I_4\};\\
G_4^{1,8}&=\{{\scriptstyle{1,1}}=I_4\};\\ G_5^{3,8}&=\{{\scriptstyle{1,2}}=I_4,\,{\scriptstyle{2,3}}=I_4,\, {\scriptstyle{3,4}}=I_4\};\\
G_6^{1,8}&=\{{\scriptstyle{1,8}}=1\},\\
G_a^{4,4}&=\{{\scriptstyle{1,1}}=I_4,\,{\scriptstyle{1,2}}=D_m^{-1}A_1,\, {\scriptstyle{1,3}}=D_m^{-1}A_2,\, {\scriptstyle{1,4}}=D_m^{-1}A_3,\, {\scriptstyle{2,2}}=I_4,\,{\scriptstyle{3,3}}=I_4,\, {\scriptstyle{4,4}}=I_4\};\\
G_b^{3,4}&=\{{\scriptstyle{1,1}}=\tau_1I_4,\,{\scriptstyle{1,2}}=-I_4,\,{\scriptstyle{2,1}}=\tau_2I_4,\, {\scriptstyle{2,3}}=-I_4,\, {\scriptstyle{3,1}}=(\tau_1+\tau_2)I_4,\,{\scriptstyle{3,4}}=-I_4\};\\
G_c^{3,4}&=\{{\scriptstyle{1,2}}=I_4,\,{\scriptstyle{2,3}}=I_4,\,{\scriptstyle{3,4}}=I_4\};\\
D_\tau&=\textrm{diag}\{\frac{e^{-2\alpha\tau_1}}{\tau_1}I_4,\, \frac{e^{-2\alpha\tau_2}}{\tau_2}I_4,\, \frac{e^{-2\alpha (\tau_1+\tau_2)}}{\tau_1+\tau_2}I_4\};\\
 R_m&=\textrm{diag}\{I_4,\, 3I_4\};\,
\tilde R_m=\textrm{diag}\{R_me^{-2\alpha\tau_1},\,R_me^{-2\alpha\tau_2},\,R_me^{-2\alpha(\tau_1+\tau_2)}\};\\
\tilde R_1&=\textrm{diag}\{R_1,\, R_1\};\,\tilde R_2=\textrm{diag}\{R_2,\, R_2\};\,\tilde R_3=\textrm{diag}\{R_3,\, R_3\};\\
\tilde R&=\textrm{diag}\{\tilde R_1,\,\tilde R_2,\,\tilde R_3\};\,
D_R=\textrm{diag}\{R_1,\, R_2,\, R_3\};\\
D_S&=\textrm{diag}\{S_1,\,S_2,\,S_3\};\,
D_e=\textrm{diag}\{e^{-2\alpha\tau_1}I_4,\,e^{-2\alpha\tau_2}I_4,\,e^{-2\alpha (\tau_1+\tau_2)}I_4\}.
\end{split}
$}
\end{equation}
Let us define
\begin{equation}
z=\begin{pmatrix}
x+D_m^{-1}A_1\int_{t-\tau_1}^tx(\mu)\textrm{d}\mu+D_m^{-1}A_2\int_{t-\tau_2}^tx(\mu)\textrm{d}\mu+D_m^{-1}A_3\int_{t-\tau_1-\tau_2}^tx(\mu)\textrm{d}\mu\\ \int_{t-\tau_1}^tx(\mu)\textrm{d}\mu\\ \int_{t-\tau_2}^tx(\mu)\textrm{d}\mu\\ \int_{t-\tau_1-\tau_2}^tx(\mu)\textrm{d}\mu
\end{pmatrix},
\end{equation}
 we can write $z=G_0\xi=G_a\xi_1$ and $\dot z=G_1\xi$. Let $P\in\mathbb{R}^{16\times 16}>0$, we  can write the following   as
\begin{equation}
\begin{array}{l}
V_1=z^TPz=\xi_1^TG_a^TPG_a\xi_1,\\
\dot V_1+2\alpha V_1=\xi^T(G_0^TPG_1+G_1^TPG_0+2\alpha G_0^TPG_0)\xi.
\end{array}
\end{equation}
Next, let the following Lyapunov Krosovskii functionals $V_2=V_{2_1}+V_{2_2}+V_{2_3}$ be expressed as
\begin{equation}
\begin{split}
V_2&=\tau_1\int_{-\tau_1}^0\int_{t+\theta}^te^{-2\alpha (t-\mu)}\dot x^T(\mu)R_1\dot x(\mu)\textrm{d}\mu\textrm{d}\theta+
\tau_2\int_{-\tau_2}^0\int_{t+\theta}^te^{-2\alpha (t-\mu)}\dot x^T(\mu)R_2\dot x(\mu)\textrm{d}\mu\textrm{d}\theta\\ &+
(\tau_1+\tau_2)\int_{-\tau_1-\tau_2}^0\int_{t+\theta}^te^{-2\alpha (t-\mu)}\dot x^T(\mu)R_3\dot x(\mu)\textrm{d}\mu\textrm{d}\theta,
\end{split}
\end{equation}
where $\alpha$ denotes the exponential delay rate. 
According to the Extended Jensen's inequality (\ref{Extend Jensen}), we can write
\begin{equation}
V_{2_1}\geq \frac{2}{\tau_1}e^{-2\alpha\tau_1}\left(\int_{-\tau_1}^0\int_{t+\theta}^t\dot x^T(\mu)\textrm{d}\mu\textrm{d}\theta\right)R_1\left(\int_{-\tau_1}^0\int_{t+\theta}^t\dot x(\mu)\textrm{d}\mu\textrm{d}\theta\right).
\end{equation}
Expanding the integral $\int_{t+\theta}^t\dot x^T(\mu)\textrm{d}\mu$ and applying the exchange of variable for integration, we can write
\begin{equation}
V_{2_1}\geq  \frac{2}{\tau_1}e^{-2\alpha\tau_1}\left(\tau_1x-\int_{t-\tau_1}^tx(\mu)\textrm{d}\mu\right)^TR_1\left(\tau_1x-\int_{t-\tau_1}^tx(\mu)\textrm{d}\mu\right).
\end{equation}
Similar inequalities can be expressed for $V_{2_2},\, V_{2_3}$. Then this yields 
\begin{equation}
V_2\geq \xi_1^TG_b^T2D_\tau D_RG_b\xi_1. 
\end{equation}
We can equivalently write $V_{2_1}$ as $V_{2_1}=-\tau_1\int_{t-\tau_1}^t(t-\mu-\tau_1)e^{-2\alpha (t-\mu)}\dot x(\mu)R_1\dot x(\mu)\textrm{d}\mu$.
Taking derivative for $V_{2_1}$ yields
\begin{equation}
\begin{split}
\dot V_{2_1}&=\dot x^T\tau_1^2R_1\dot x-\tau_1\int_{t-\tau_1}^t\frac{\partial}{\partial t}\big((t-\mu-\tau_1)e^{-2\alpha (t-\mu)}\dot x^T(\mu)R_1\dot x(\mu)\big)\textrm{d}\mu\\
&=\dot x^T\tau_1^2R_1\dot x-\tau_1\int_{t-\tau_1}^te^{-2\alpha (t-\mu)}\dot x^T(\mu)R_1\dot x(\mu)\textrm{d}\mu-2\alpha V_{2_1}.
\end{split}
\end{equation}
Then the following inequalities hold
\begin{equation}\label{V2}
\begin{array}{l}
\dot V_{2_1}+2\alpha V_{2_1}\leq \dot x^T\tau_1^2R_1\dot x-\tau_1e^{-2\alpha\tau_1}\int_{t-\tau_1}^t\dot x^T(\mu)R_1\dot x(\mu)\textrm{d}\mu,\\
\dot V_{2_2}+2\alpha V_{2_2}\leq \dot x^T\tau_2^2R_2\dot x-\tau_2e^{-2\alpha\tau_2}\int_{t-\tau_2}^t\dot x^T(\mu)R_2\dot x(\mu)\textrm{d}\mu,\\
\dot V_{3_1}+2\alpha V_{3_1}\leq \dot x^T(\tau_1+\tau_2)^2R_3\dot x-(\tau_1+\tau_2)e^{-2\alpha (\tau_1+\tau_2)}\int_{t-\tau_1-\tau_2}^t\dot x^T(\mu)R_3\dot x(\mu)\textrm{d}\mu.
\end{array}
\end{equation}
Substituting the right hand-side of (\ref{state space}), where both sides of the equation is multiplied by $D_m^{-1}$ and $B_s\overline F_d$ is replaced by $B_d\overline \gamma F_d$, for $\dot x$  in (\ref{V2}) and applying Wirtinger inequality in Lemma \ref{Wintinger} for the integral terms, yields
\begin{equation}
\begin{array}{l}
\dot V_2+2\alpha V_2\leq \xi^T\Big(G_2^T(\tau_1^2R_1+\tau_2^2R_2+(\tau_1+\tau_2)^2R_3)G_2-G_3^T\tilde R\tilde R_mG_3\Big)\xi.
\end{array}
\end{equation}
Next, let us define 
\begin{equation}\begin{split}
V_3&=\int_{t-\tau_1}^te^{-2\alpha (t-\mu)}x^T(\mu)S_1x(\mu)\textrm{d}\mu+\int_{t-\tau_2}^te^{-2\alpha (t-\mu)}x^T(\mu)S_2x(\mu)\textrm{d}\mu\\ &+\int_{t-\tau_1-\tau_2}^te^{-2\alpha (t-\mu)}x^T(\mu)S_3x(\mu)\textrm{d}\mu.
\end{split}\end{equation}
 Using Jensen's inequality (\ref{Jensen}) and taking derivative we get
\begin{equation}
\begin{array}{l}
V_3\geq \xi_1^TG_c^TD_\tau D_sG_c\xi_1,\\
\dot V_3+2\alpha V_3\leq \xi^T\Big(G_4^T(S_1+S_2+S_3)G_4-G_5^TD_eD_sG_5\Big)\xi.
\end{array}
\end{equation}
Let us define
\begin{equation}\label{LMI1}
\begin{split}
\Omega&=G_a^TPG_a+G_b^T2D_\tau D_RG_b+G_c^TD_\tau D_SG_c>0,\\
\Theta&=G_0^TPG_1+G_1^TPG_0+2\alpha G_0^TPG_0+G_2^T\big(\tau_1^2R_1+\tau_2^2R_2+(\tau_1+\tau_2)^2R_3\big)G_2-G_3^T\tilde R\tilde R_mG_3\\&+G_4^T(S_1+S_2+S_3)G_4-G_5^TD_eD_SG_5-bG_6^TG_6<0,
\end{split}
\end{equation}
we formulate our main results as follows: 
\begin{theorem}\label{MFRC}
Given control parameters $\bar m$, $g,\,p$, $\lambda_1,\,\lambda_2,\,\omega_{c_1},\,\omega_{c_2}\geq 0$ and 
 constant delays $\tau_1,\, \tau_2>0$, and tuning parameters $\alpha\geq 0,\, b\geq 0$,  if there exist $16\times 16$ matrix $P>0$, $4\times 4$ matrices $R_1,\,R_2,\,R_3> 0$ and $4\times 4$  matrices $S_1,\,S_2,\,S_3>0$, such that 
LMIs in (\ref{LMI1}) hold, then the RC system (\ref{state space}) is exponentially attracted by the ultimate bound for $t\rightarrow \infty$
\begin{equation}\label{ulti_bound_x1}
|\dot x_2|\leq \|x\|\leq  M|F_d|,
\end{equation}
where $M=\sqrt{\frac{b}{2\alpha \underline{\lambda}(\Omega)}}\overline \gamma$, and $\underline{\lambda}(\Omega)$ denotes the smallest eigenvalue of $\Omega$.
\end{theorem}
\begin{proof}Let $V=V_1+V_2+V_3$, it follows $V\geq \xi_1^T\Omega\xi_1\geq \underline{\lambda}(\Omega)\|x\|^2$. Then $\dot V+2\alpha V-b(\overline \gamma F_d)^T(\overline \gamma F_d)\leq \xi^T\Theta\xi<0$ implies (\ref{ulti_bound_x1}) according to ISS criteria in Lemma \ref{input-to-state}.
\end{proof}

The following summarizes the design procedure for controller synthesis:
\begin{enumerate}
\item Compute $\beta_i$ in $\lambda_i$ that  minimizes the bound on the RC-induced-disturbance $\overline \gamma$   for  given choices of    $\omega_{c_i}$, $\tau_i$ and $\omega_{d_x}$ using LMIs in (\ref{LMI_gamma}). 
\item Tune the control parameters $g$, $p$ such that LMIs in (\ref{LMI1}) are feasible. Otherwise, restart from step 1) by 
increasing $\overline \gamma$ with smaller $\lambda_i$.  
\end{enumerate}

\section{\textbf{Experimental results}}
Fig. \ref{setup} shows an active vibration isolation system (AVIS) \cite{Tjepkema}, which is a hard-mount system with suspension frequency related to the  vertical translation mode of 24 Hz. 
The active system is used to suppress both indirect disturbances from floor vibrations and direct disturbances acting on the payload. 
Six voice-coil motor (VCM) actuators are connected to the payload via wire springs, forming a configuration resembling a Stewart platform. Accelerometers are attached to the Stewart platform to measure platform accelerations in the actuator directions. 
For damping control, acceleration signals are integrated to obtain velocity signals using weak integrators \cite{Beijen}, \cite{Sil potential}, \cite{Wouter filter}.
We only consider direct disturbance by applying wide-band noise with noise power of $2^{-4}\frac{\textrm{V}}{\sqrt{\textrm{Hz}}}$ to the actuators. We perform a coordinate transformation to isolate the payload's vertical (z) direction using active control. 
 The controllers are implemented on a digital signal processor running at a sampling frequency of $f_s=5000$ Hz.
 The mass, stiffness, and damping parameters have been identified as $m=7.6\,\textrm{kg}$, $k=180210\,\textrm{N/m}$, and $d=58.2\,\textrm{Ns/m}$, respectively. 
 
 Since we can measure the applied noise profile to the VCM actuator and the acceleration signal from the accelerometers for the payload, we can calculate the frequency response transfer function from these two measurement signals using spectral analysis as defined in  \cite{Pintelon}. 
Fig. \ref{sparse} shows that the measured open-loop system (actual plant) matches the plant model $G$ relatively accurately up to 700 Hz, where the parasitic dynamics of the platform suspension kick in.
For illustration, in the figures below, we denote $F_{d_C}=F_d+F_R$ as the compensated disturbance, where $F_R$ denotes the RC control in (\ref{F1_phi}). Since we can measure both $F_{d_C}$ and $F_d$, we can plot the measured compensated disturbance by RC using $\frac{F_{d_C}}{F_d}$. This quantity reveals the disturbance characterization due to RC.  The original system (\ref{linearSys})  can be written as $m\dot x_2=-kx_1-dx_2+F_{d_C}+F_N$.
\begin{figure}
\centering
  \includegraphics[width=.7\linewidth]{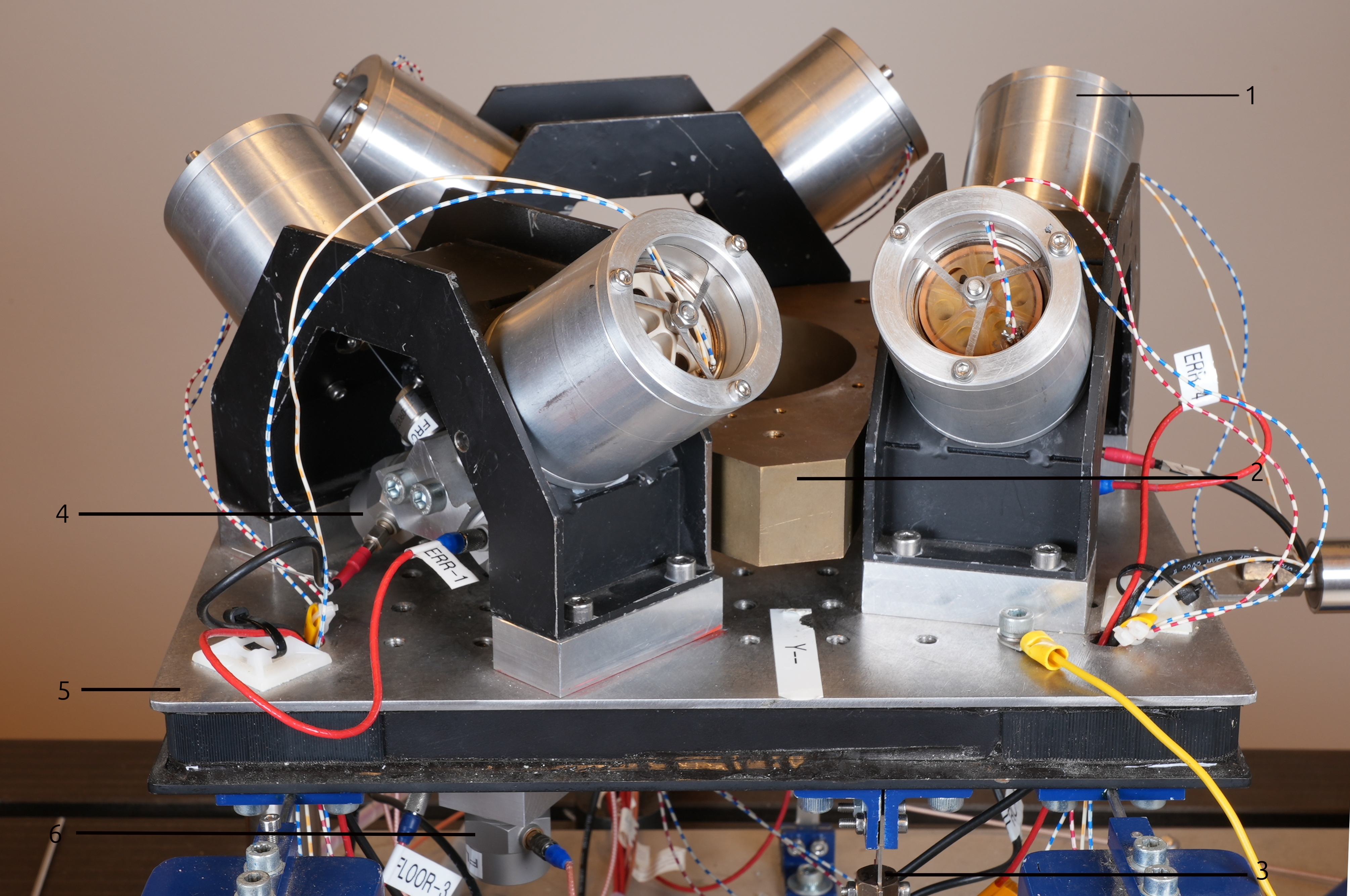}
  \caption{Test setup with the brass platform suspended to the cylindrical actuators, which are connected to the floor plate via the black bridge structure. Floor plate accelerations can be measured and excited using piezo actuators. 1: VCM actuator, 2: Payload, 3: Piezo electric actuators for floor excitation, 4: Payload accelerometers, 5: Floor, 6: Floor accelerometer 
  }
  \label{setup}
\end{figure}

In the following, we first show the experimental results for a single-period RC designed with an LPF (Low Pass Filter) and a BPF (Band Pass Filter) to verify the conservativeness of our Theorem. 
Secondly, we demonstrate the robustness of the proposed MPRC (RC+TDC) relative to the existing robust RC scheme under frequency uncertainties.  
Thirdly, we compare the proposed MPRC's performance with that of the existing cascaded and parallel MPRC schemes. 
Fourthly, we show optimization of the proposed MPRC at a non-repetitive frequency. 
Finally, we show that the combination of TDC and RC can be used to identify unknown periodic disturbances with higher accuracy than TDC alone, both in phase and magnitude. By utilizing an existing frequency identifier from the literature, the frequency of the estimated disturbance can be identified, which is then used in RC to achieve increased robustness against repetitive frequency uncertainty. 
Since we consider both a low-pass filter and a band-pass filter, we explain the notation of the cutoff frequencies as follows.
For the LPF case, $\omega_{c_1}$ and $\omega_{c_2}$ denote the low-pass cutoff frequencies of two LPFs in a two-period RC scheme, or in a one RC with one TDC scheme. 
For BPF case, $\omega_{c_1}$ and $\omega_{c_2}$ are used however to denote the low-pass cutoff frequency $\omega_{c_{1_i}}$ and the high-pass cutoff frequency $\omega_{c_{2_i}}$, respectively, of a single BPF in a single RC or a single TDC scheme.
\subsection{Single frequency RC with low-pass and band-pass filters} 
 Firstly, let us assess the level of conservatism of our result. Given the controller and filter parameters in the first row of Table \ref{tab:table1}, we want to choose a $\lambda$ value that produces maximum RC action. To achieve that, we find $\lambda=1.1792$ that yields $\beta=\gamma=0$. However, our result in LMIs (\ref{LMI1}) is not feasible with this $\lambda$  value. 
We can assess the conservatism of our result in this way. Because for single RC $n=1$, we have 3 by 3 matrices $A,\,A_1$. Then $A_{31}+A_{1_{31}}=-(k+p)\omega_{c_1}+\omega_{c_1}\lambda_1 k=0$ for $p=0$ and $\lambda_1=1$.
   Since we have chosen $p=0$, the system matrix $A+A_1$ becomes unstable for $\lambda_1>1$. This indicates that our LMI feasibility is not conservative in this case, as $\lambda_{\textrm{Theorem}}=0.997$ is very close to the stability boundary condition. It is, however, possible to increase $\lambda$ beyond unity if we choose a significant  positive control parameter $p$. Since there is no position sensor mounted on the setup, we have not verified this case. Although acceleration can be integrated twice to obtain a position signal, achieving an absolutely accurate position is challenging in the presence of sensor noise and undefined initial conditions.  

We  evaluate controller performance for the two cases in Table \ref{tab:table1} below.
 Parameters in Table \ref{tab:table1} are chosen to satisfy stability conditions  in Theorem \ref{MFRC}. 
We denote $|\frac{\dot x_2}{F_d}|_m$ 
as the measured transfer function in the experiments and $|\frac{\dot x_2}{F_d}|_e=M$ as the estimated bound in (\ref{ulti_bound_x1}).

\begin{table}[h!]
  \begin{center}
    \caption{1 RC. $\bar m=7$, $g=1660$, $p=0$, $\omega_d=50$ Hz}
    \label{tab:table1}
\resizebox{\textwidth}{!}{
    \begin{tabular}{|c|c|c|c|c|c|c|c|c|} 
      \hline
     Filter&$\tau$ (s)&max $[\lambda_{\textrm{Theorem}},\,\lambda_{\textrm{Testbed}}]$ &  $\omega_{c_1}$ (Hz) &$\omega_{c_{2}}$ (Hz)& $|\frac{\dot x_2}{F_d}(\omega_d)|_{m}$ (dB) & $|\frac{\dot x_2}{F_d}(\omega_d)|_{e}$ (dB)         \\
     \hline
      LPF&  0.0185&0.997, 0.86&80&0 &-29.4&     -4.59   \\
      \hline
BPF&0.02&1.318, 1.318&80&25&-48&  -35.9 \\
     \hline      
    \end{tabular}
}
  \end{center}
\end{table}
 \begin{figure}
\centering
  \includegraphics[width=1\linewidth]{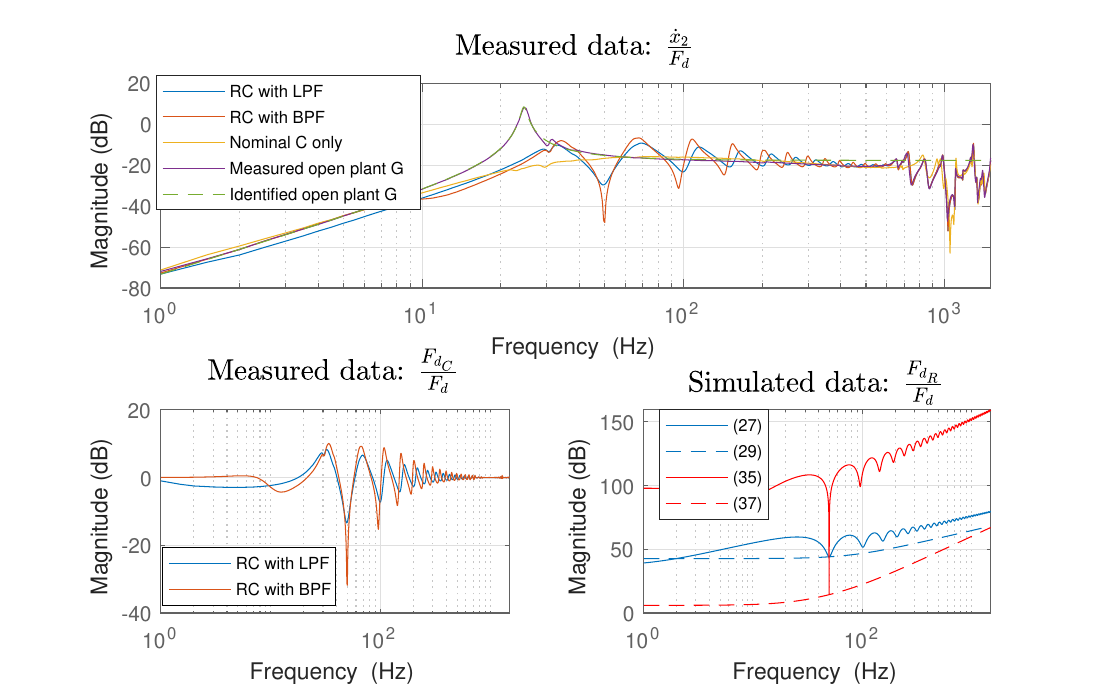}
  \caption{Performance in terms of the transfer function from $F_d$ to $\dot x_2$ for a single-period RC with an LPF and a BPF. The bottom-right subfigure plots the simulated transfer functions of the RC-induced-disturbances $\frac{F_{d_R}}{F_d}$ in equation (\ref{remark7_1}) for LPFRC with blue line, in equation (\ref{band_1}) for BPFRC with red line, and their minimized value in terms of using $\overline \gamma$ obtained at $\omega_{d_x}=50$Hz,  (\ref{eq:gamma_bar_a}) in yellow dashed line and (\ref{eq:gamma_bar_b}) in purple dashed line, respectively.
The bottom-left subfigure shows experimentally measured data of the compensated disturbance with and without RC. The upper subfigure shows the experimentally measured transfer function of acceleration to the disturbance (\ref{Fd_x2}).   
  }
  \label{sparse}
\end{figure}
For a single RC with a LPF, we can theoretically choose a maximum $\lambda=0.997$ to attain feasibility of the LMIs (\ref{LMI1}) in Theorem \ref{MFRC}, but in the experiment, we can only choose up to $\lambda=0.86$ before the system becomes unstable. This is mainly due to the destabilizing effect of the parasitic dynamics after 700Hz.  With parameters 
given by the first row of Table \ref{tab:table1},
 we choose   $\omega_{d_x}=\omega_d=50$Hz in the LMIs  (\ref{LMI_gamma}) and obtain  $\beta=160.47$, $\gamma=0.2707$ and $\overline \gamma=44.108$ dB for the given $\lambda_{\textrm{Testbed}}$ value.
Then we find the feasibility of LMIs in (\ref{LMI1}) with
 $\alpha=1.596$, $b=8\times 10^{-11}$, which yields $\underline{\lambda}(\Omega)=1.8\times 10^{-6}$. This allows to calculate $M$ for $|\frac{\dot x_2}{F_d}|_e$.

For the BPF case, we would get a perfect disturbance cancellation ($\gamma=F_{d_R}=0$) if we choose $\omega_{d_x}=50$Hz with $\lambda=1.318$. We have applied this optimum value in the experiment. 
  We  choose $\omega_{d_x}=50.00018$ Hz in the LMIs  (\ref{LMI_gamma}) which is  close to the disturbance frequency $\omega_d=50$ Hz.
   We then obtain $\beta=-8.7\times 10^{-11}$, $\gamma=2.5337\times 10^{-5}$ and $\overline \gamma=14.44$ dB.
LMI (\ref{LMI1})  is feasible with  
  $\alpha=5$, $b=2\times 10^{-11}$, which gives  $\underline{\lambda}(\Omega)=2.16\times 10^{-7}$ and hence the obtained $|\frac{\dot x_2}{F_d}|_e$. The larger $\alpha$ in this case implies a larger exponential decay rate.
The estimated bound on the output transfer function $|\frac{\dot x_2}{F_d}|_e$  at the target frequency $\omega_d$ for the LPF case shows more conservatism than the BPF case. This is because the optimization spans a larger frequency range. The LPF case clearly performs better over low frequencies below 10 Hz.

The bottom-left subfigure  in Figure \ref{sparse} shows the compensated disturbance $\frac{F_{d_C}}{F_d}$, using experimentally measured data for the LPF case, the BPF case, and the case using only nominal control without RC. This shows RC attenuates disturbance at the desired frequency at the expense of amplification at intermediate frequencies.
BPFRC exhibits greater attenuation at 50 Hz but greater amplification than LPFRC at intermediate frequencies. This can be explained by the RC-induced disturbance by BPFRC being lower at the target frequency but higher at intermediate frequencies, as shown in the bottom-right subfigure. 
In the bottom-right subfigure of  figure \ref{sparse},
we simulate the RC-induced disturbance $\frac{F_{d_R}}{F_d}$ in equation (\ref{remark7_1}) for LPFRC, and in equation (\ref{band_1}) for BPFRC, and then simulate their minimized value in terms of the minimized $\overline \gamma$ obtained at 50 Hz in (\ref{eq:gamma_bar_a}) and in (\ref{eq:gamma_bar_b}), respectively.
The RC-induced disturbance by BPFRC in the red line is less than that by LPFRC in the blue line at  50 Hz, but is higher at other frequencies. 
The upper-half subfigure  shows the performance of acceleration against disturbance with nominal control as well, where $ F_N$ gives the nominal control in $m\dot x_2=-kx_1-dx_2+F_{d_C}+F_N$. 
The nominal control is used to damp the mechanical resonance, defined by system parameters $k,\,d$, as well as the compensated disturbance $F_{d_C}$.  The output transfer functions $\frac{\dot x_2}{F_d}$ show that the RC has better performance than nominal control at target frequencies but is less favourable at intermediate frequencies.

\subsection{Robust RC}
This section shows the improvement of the proposed approach in the  robustness of RC to frequency uncertainties at the designed RC frequency by combining RC with TDC, the first row in Table \ref{tab:table2}. This is then compared  with the robust high order RC proposed in \cite{Steinbuch}, which is essentially a two MPRC design for a single period disturbance, as explained in Remark \ref{remark1},   the second row in Table \ref{tab:table2}. The cutoff frequencies $\omega_{c_1}$, $\omega_{c_2}$ for both LPFs are set at  80 Hz.
\begin{table}[h!]
  \begin{center}
    \caption{Robust RC. $\bar m=7$, $g=1660$, $p=0$, $\omega_{d_1},\omega_{d_2}=50$ Hz, $ \omega_{c_1},\omega_{c_2}=80$ Hz}
    \label{tab:table2}
    \begin{tabular}{|c|c|c|c|c|} 
      \hline
    Filter& Case &$\tau_1,\,\tau_2$ (s)& max $[\lambda_1,\,\lambda_2]_{\textrm{Theorem}}$ &$[\lambda_1,\,\lambda_2]_{\textrm{Testbed}}$\\
      \hline
LPF& 1 RC+1 TDC&  0.0182, 0.0002&1.0013, 1 & 0.8, 0.4 \\
     \hline
    LPF&  2 RC&  0.0182, 0.0182&  0.84, 0.84 & 0.6, 0.6 \\
      \hline      
    \end{tabular}
  \end{center}
\end{table}

In the experiment, we choose RC gains such that $\lambda_1+\lambda_2$ is the same in both cases for fair comparison. The RC gains used in the testbed are all below their maximum, which are obtained by Theorem \ref{MFRC}. Larger than these values will cause instability due to parasitic dynamics.
Note, due to TDC in the second loop, we can choose $\lambda_1>1$ in Theorem \ref{MFRC} in this case, compared to the single period RC with a low-pass filter in the previous section, where we could only choose it to be 0.997.
In Figure \ref{fig:3}, we plot the compensated disturbance in the bottom subfigures and the transfer function of acceleration to disturbance in the upper subfigure.  The transfer function of acceleration to applied noise is obtained by damping the compensated disturbance using nominal control. 
In the figure,
both (1 RC+1 TDC) and 2RC (2 RC) yield almost identical values at 50 Hz, which are lower than those with LPF, which is designed using the parameters given in the first row of Table \ref{tab:table1}. Both methods allow widening the valley at 50Hz rather than just 1 RC. However, 2RC induces the most significant disturbance at all the intermediate frequencies. 
1RC1TDC yields a lower disturbance at 50 Hz than 1 RC, but a higher disturbance at intermediate frequencies above 100 Hz, due to the same cut-off frequencies of the two filters at 80 Hz.
 There are small regions between 42Hz and 56Hz where 2RC produces lower disturbance than  1RC1TDC, as shown in the bottom-right subfigure. Apart from these frequencies, 1RC1TDC produces less RC-compensated disturbance than 2RC, including the intermediate frequencies at around 30 Hz and 70 Hz.
Summarizing, high-order RC proposed in \cite{Steinbuch} (2RC in this case) performs slightly better in the small neighbourhood of the RC frequency, but the proposed 1RC1TDC scheme performs better for the rest of the frequencies,  with significantly less amplification in the intermediate frequencies. 
  
\begin{figure}[H]
\centering
  \includegraphics[width=1\linewidth]{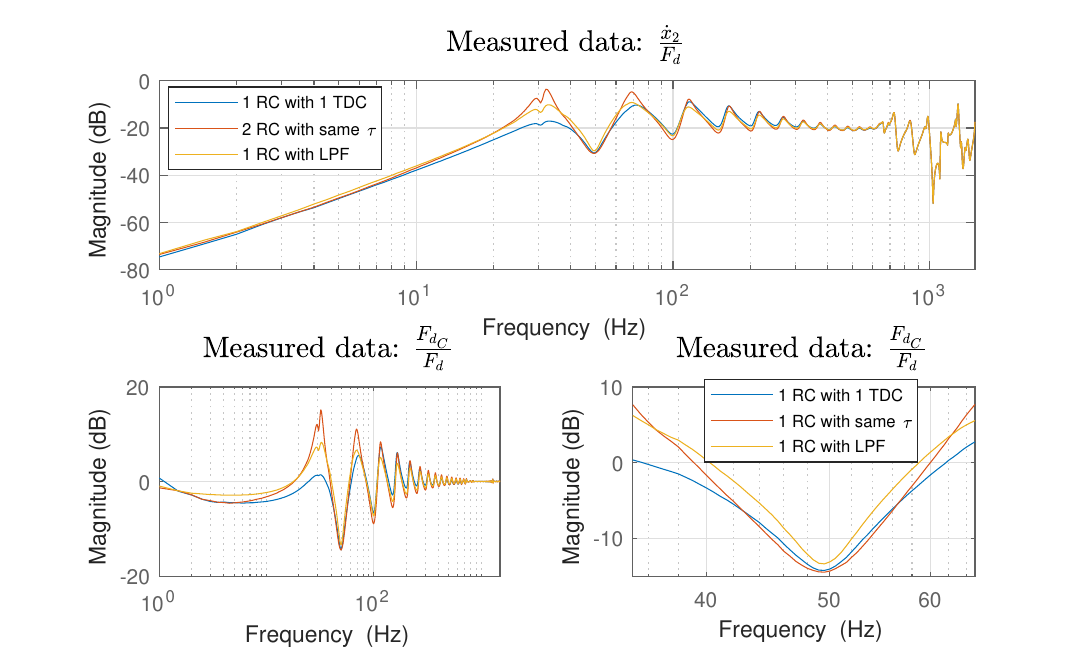}
  \caption{Comparison of RC+TDC with high order RC \protect\cite{Steinbuch} in robustness against frequency uncertainties. The bottom subfigures show experimentally measured data of the compensated disturbance $\frac{F_{d_C}}{F_d}$. The upper subfigure shows the experimentally measured transfer function of acceleration to the applied noise.}
  \label{fig:3}
\end{figure}

\subsection{Multi-period RC}

\subsubsection{Proposed 2 RC}\label{2RC_sec}
\begin{table}[h!]
  \begin{center}
    \caption{2 RC, $\bar m=5$, $g=3000$, $p=0$, $\omega_{d_1}=20$ Hz, $\omega_{d_2}=50$ Hz}
    \label{tab:table3_1}
    \begin{tabular}{|c|c|c|c|c|c|} 
      \hline
  Filter&   Case &$\tau_1,\,\tau_2$ (s)&max $[\lambda_1,\,\lambda_2]_{\textrm{Theorem}}$& max $[\lambda_1,\,\lambda_2]_{\textrm{Testbed}} $ &  $\omega_{c_1},\,\omega_{c_2}$ (Hz)        \\
     \hline
    LPF&  1& 0.0463, 0.0182&0.86, 0.9& 0.75, 0.75&40, 80  \\
      \hline
BPF& 2&  0.05, 0.0188&---& 1.25, 1.2& $\begin{matrix}\textrm{40, 10}\\ \textrm{80, 10}\end{matrix} $   \\
      \hline      
    \end{tabular}
  \end{center}
\end{table}
Table \ref{tab:table3_1} shows 2 RC in two cases. Disturbance frequencies are at 20 and 50 Hz.
Case 1 has one low-pass filter for each RC with low-pass cutoff frequencies at $\omega_{c_1}=$40 Hz and $\omega_{c_2}=$80 Hz, respectively. The maximum RC gains $\lambda_1,\,\lambda_2$  are obtained via stability conditions in LMIs (\ref{LMI1}) in Theorem \ref{MFRC}.  The maximum experimental RC gains, however, cannot be chosen as large as their theoretical maximums due to parasitic dynamics. 
Case 2 considers one band-pass filter for each RC. 
Theorem \ref{MFRC} is, however, not applicable to this case as it only considers one rather than two band-pass filter designs.
 For the first disturbance at 20 Hz, we design the low-pass cutoff frequency and high-pass cutoff frequency for the first BPF as $\omega_{c_1}=$40 Hz and $\omega_{c_2}=$10 Hz, respectively.
For the second disturbance at 50 Hz, they are designed for the second BPF to be $\omega_{c_1}=$80 Hz and $\omega_{c_2}=$10 Hz.
The RC gains in Case 2 are chosen to be equal to  $[\lambda_1,\lambda_2]_{\gamma_1,\gamma_2=0}=[1.25,\,1.2]$ for theoretical maximum attenuation $\gamma_1,\,\gamma_2=0$ as obtained from LMI tool (\ref{LMI_gamma}) for the BPF case, by setting $\omega_{d_x}=20$ Hz and $\omega_{d_x}=50$ Hz to be equal to the target frequencies for the first and second RC, respectively.  
Fig. \ref{2RC_1} shows that Case 1 with LPF performs better at low frequencies below 5 Hz and at the intermediate frequency at 30 Hz of the first RC and the intermediate frequency at 65 Hz of the second RC.
Case 2 with BPF, however, yields more attenuation at the designed fundamental frequencies $\omega_{d_1}=20$ Hz, $\omega_{d_2}=50$ Hz. 
Comparing the two cases, the first case, which can be proved to be stable  by the theorem, yields an
averaged performance with less suppression at the RC disturbance frequencies, but also less amplification at intermediate frequencies.
The RC gains in Case 2, however, are chosen for maximum performance without a stability check. This leads to improved performance, particularly at the fundamental RC disturbance frequencies.

\begin{figure}[H]
\centering
  \includegraphics[width=1\linewidth]{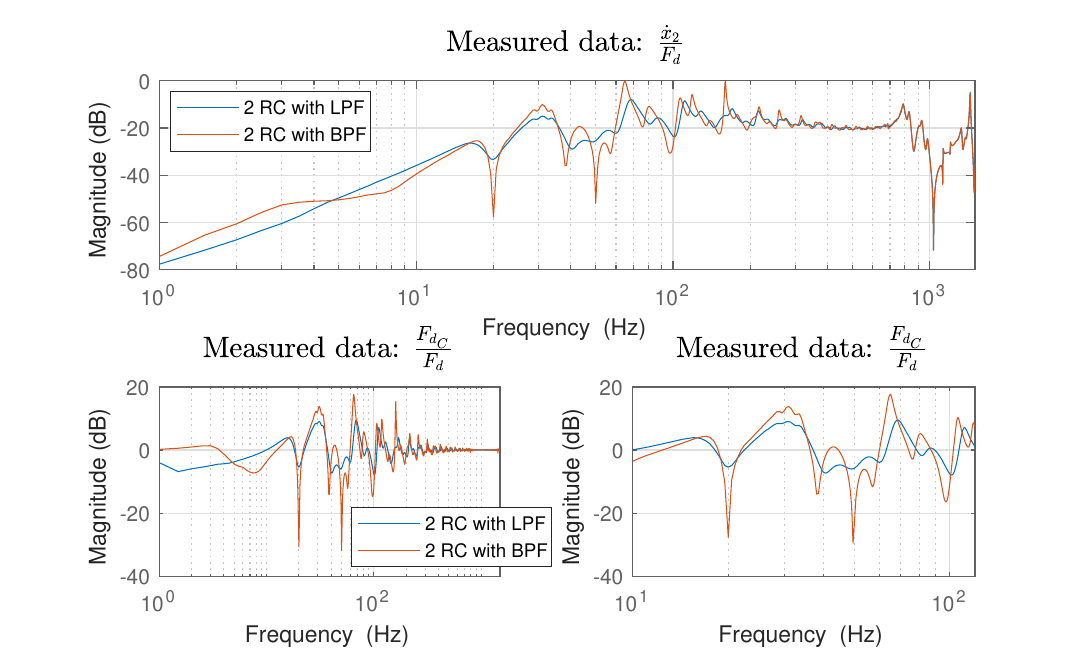}
  \caption{Proposed MPRC with LPF and BPF  for RC frequencies at 20 Hz and 50 Hz. The bottom subfigures show experimentally measured data of the compensated disturbance. The upper subfigure shows the experimentally measured transfer function of acceleration to the applied noise.
    }
\label{2RC_1}
\end{figure}

\subsubsection{Comparison to parallel and cascaded MPRC}\label{parallel}
We consider 2 BPFRC using parallel MPRC or cascaded MPRC and compare them with the proposed approach. The cascaded MPRC is implemented using the identified plant model.
Parameters are given in  Table \ref{tab:table5}.
It has been tested that before instability occurs, the maximum RC gains can be chosen as 0.97 and the maximum nominal control gain (or so-called 'skyhook gain' in \cite{Sil potential}, \cite{Wouter filter}) as  1.5 under parallel RC or cascaded RC, which is very small. 
As discussed in Remark \ref{remark2}, small loop transfer function $L=CG$  in cascaded RC will not eliminate  multi-frequency interaction. 
 In comparison, the proposed approach allows choosing larger RC gains (1.25 and 1.2) and a significantly larger skyhook gain (3000). 
The performance comparison is given in Fig. \ref{2RC_2noC}. 
It is shown that both parallel and cascaded RC lead to significant amplification at 20 Hz, even though this is the target frequency of the first RC. 
Cascaded RC  amplifies relatively less than parallel RC at this frequency, but it still fails to suppress the disturbance at this frequency. 
\begin{table}[h!]
  \begin{center}
    \caption{Parallel or cascaded RC, $\bar m=1$, $p=0$, $\omega_{d_1}=20$ Hz, $\omega_{d_2}=50$ Hz}
    \label{tab:table5}
    \begin{tabular}{|c|c|c|c|c|} 
      \hline
  Filter&   Case & skyhook gain $g$ & max $[\lambda_1,\,\lambda_2]_{\textrm{Testbed}} $ &  $\omega_{c_1},\,\omega_{c_2}$ (Hz)        \\
      \hline
BPF& Parallel \& Cascaded&1.5 &0.97, 0.97& $\begin{matrix}\textrm{40, 10}\\ \textrm{80, 10}\end{matrix} $ \\
      \hline  
      BPF& Case 2 in Table \ref{tab:table3_1} &3000&1.25, 1.2& $\begin{matrix}\textrm{40, 10}\\ \textrm{80, 10}\end{matrix} $ \\
\hline
    \end{tabular}
  \end{center}
\end{table}
\begin{figure}[h!]
\centering
  \includegraphics[width=1\linewidth]{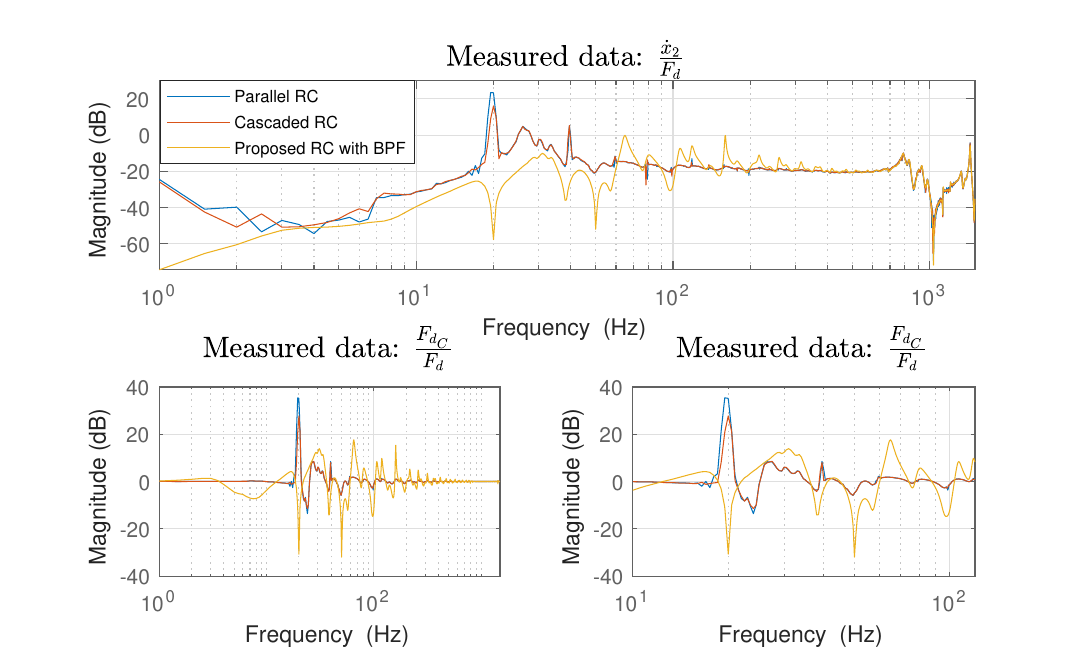}
  \caption{Comparison of parallel MPRC and cascaded MPRC with the proposed MPRC. The bottom subfigures show experimentally measured data of the compensated disturbance. The upper subfigure shows the experimentally measured transfer function of acceleration to the applied noise. }
  \label{2RC_2noC}
\end{figure}

In comparison, the proposed MFRC leads to desirable suppression at its designed target frequencies of 20 Hz and 50 Hz, as well as their harmonics, despite the use of a significantly high nominal control gain.  The significantly larger nominal control in the proposed RC allows suppression of the system's resonance at 25 Hz and increased RC gains at the target frequencies.  
We note that we do not intentionally choose a smaller skyhook gain for the Parallel and Cascaded RC in favour of a comparison to the proposed approach, which uses a significantly larger skyhook gain. The values of $\lambda_1$ and $\lambda_2$ in the Parallel and Cascaded RC are also smaller than those in the proposed design; however, this does not allow a larger skyhook gain to be chosen. 
It is not necessary to choose such a large nominal control gain in the proposed scheme if the RC gains are chosen smaller. We show that the proposed scheme allows for an increase in the nominal control gain that is much larger than that of Cascaded MPRC. This reveals that both Parallel RC and Cascaded RC may exhibit stability issues when using a larger nominal control gain, which is needed to eliminate multi-frequency interactions.

 \subsubsection{Proposed 2 RC with optimization at non-repetitive frequency}RC with TDC using the identified frequency. The main figure shows experimentally measured data of the compensated disturbance for each frequency. The subplot shows the experimentally measured transfer function of acceleration to the applied  disturbance at each frequency.
 We first design a multi-period RC for two fundamental disturbance frequencies at 50 Hz and 140 Hz, as shown in the first case of Table \ref{tab:table3}. Then we optimize the RC gains $\lambda_1,\,\lambda_2$ so that, in addition to the two fundamental disturbance frequencies, suppression at a non-repetitive frequency of 153 Hz is optimized, as in the second case in Table \ref{tab:table3}.  
 Both cases use a low-pass filter for each RC, allowing the stability condition given by LMI (\ref{LMI1}) to be applied to both. 
For the first case,
we choose the $\lambda$ values in the LMI (\ref{LMI_gamma}) with $\beta_i=0.5$, $\beta_j=0$ for maximum suppression at $\omega_{d_x}=50$ Hz and $\omega_{d_x}=140$ Hz. 
This yields the transfer function of $\frac{\dot x_2(s)}{F_d(s)}$ given by the blue line in Fig. \ref{fig:5}.

It is now of interest to improve suppression at 153 Hz in addition to 50 Hz and 140 Hz.
Our objective is to minimize the disturbance at $\omega_{d_x}=153$ Hz by minimizing $r_i$, $r_j$ in (\ref{LMI_gamma}).
We chose the smallest $r_i=0.4$ before LMI (\ref{LMI_gamma}) becomes infeasible, this yields $\beta_i=-340.8$, $\gamma_i=0.397$, $\lambda_i=1.66076$. We chose the smallest 
   $r_j=0.43$ and this yields
$\beta_j=171.7$,  $\gamma_j=0.4294$, $\lambda_j=1.0204$. 
We chose $g=3000$  for the LMIs in (\ref{LMI1}) to be feasible with the $\lambda$ values. The disturbance at 153 Hz is now suppressed more, at the expense of less suppression at 50 Hz and 140 Hz, as indicated by the red line in the figure.
\begin{table}[h!]
  \begin{center}
    \caption{2 LPFRC without/with optimization at non-repetitive frequency at 153 Hz. $\bar m=8$, $\omega_{d_1}=50$ Hz, $\omega_{d_2}=140$ Hz}
    \label{tab:table3}
    \begin{tabular}{|c|c|c|c|c|} 
      \hline
     Case & $g$ &$\tau_1,\,\tau_2$ (s)& \textrm{max} $[\lambda_1,\,\lambda_2]_{\textrm{Theorem}}$&  $\omega_{c_1},\,\omega_{c_2}$ (Hz)        \\
     \hline
      Without optimization& 1660& 0.0185, 0.0066&1.11723, 1.11803 &100, 280  \\
      \hline
With optimization & 3000& 0.0185, 0.0066&1.66076, 1.0204 &100, 280  \\
      \hline      
    \end{tabular}
  \end{center}
\end{table}
Compared to the first case, the significant $\lambda_1$ value causes less suppression at the first disturbance fundamental frequency of  50 Hz, but more suppression at 153 Hz. The decreased $\lambda_2$ leads to less suppression at the target frequency of 140 Hz, but the amplification of the RC-induced disturbance at 153 Hz is also reduced as much as possible. Therefore, the combination of the new set of $\lambda_1,\,\lambda_2$ enables 
 for an improved suppression at 153 Hz at the cost of reduced suppression at the two disturbance fundamental frequencies.

\begin{figure}
\centering
  \includegraphics[width=.7\linewidth]{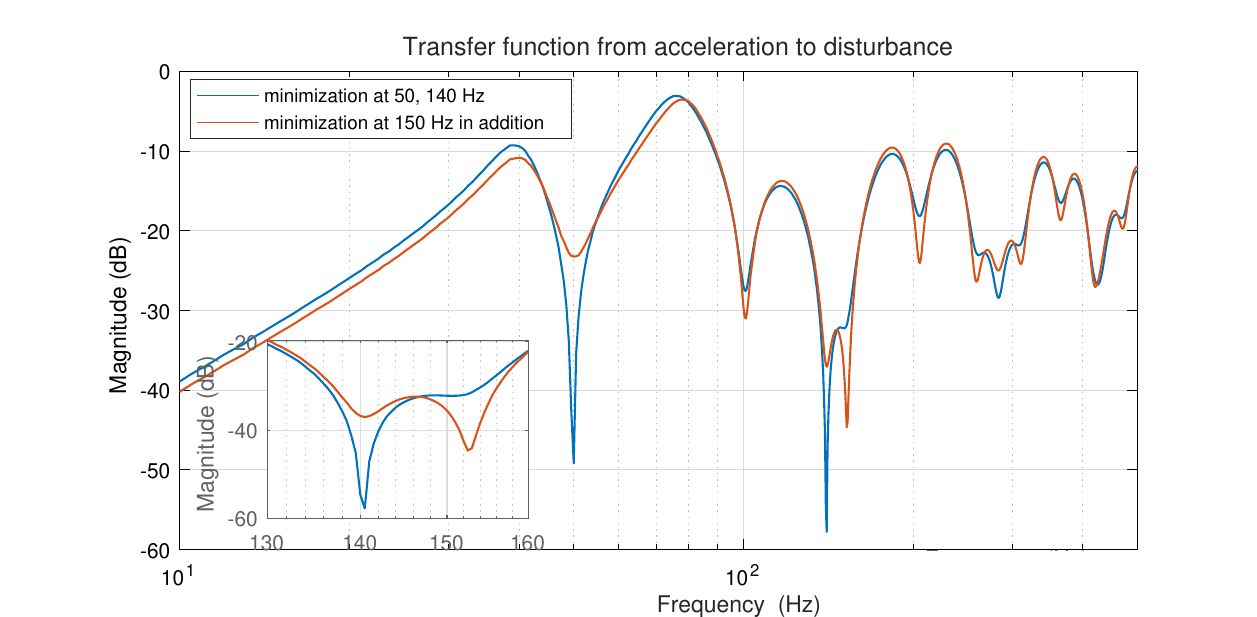}
  \caption{MPRC at 50 Hz and 140 Hz, without (blue) and with (red) considering optimization at an non-repetitive frequency at 153 Hz. The main figure shows the experimentally measured transfer function of acceleration to the applied noise. The subplot shows the zoom-in of the main figure.  }
  \label{fig:5}
\end{figure}

\subsection{RC with unknown varying frequency}\label{varying freuquency}
In this section, we demonstrate that the proposed scheme enables the estimation of an unknown periodic disturbance using TDC in the first loop, as discussed in Remark \ref{TDC remark}. We then obtain the frequency of the estimated disturbance with the assistance of an adaptive  Frequency Identifier (FI) \cite{Qing}.
The problem of unknown frequency estimation has been studied extensively using various techniques, including adaptive notch filtering, extended Kalman filter frequency estimation, the phase-locked loop technique, and adaptive identification \cite{Kang}--\cite{X.Xia}.
We select the specific FI for its simple structure, which makes implementation easier. A detailed implementation form is available in the reference.  
 The identified frequency is then
implemented as a delay for RC in the second loop. 
Control parameters and FI gain are given in Table \ref{tab:table7}.
\begin{table}[h!]
  \begin{center}
    \caption{ 1 TDC, 1 RC, 1 frequency identifier, $\bar m=3$, $g=400$ $p=0$, $\omega_{d_1}=33.3\rightarrow 25 \rightarrow 20$ Hz is unknown}
    \label{tab:table7}
    \begin{tabular}{|c|c|c|c|c|c|} 
      \hline
    Filters& Methods & $\tau_1$ (s)&$\tau_2$ (s)& $\lambda_1,\,\lambda_2$&  $\omega_{c_1},\,\omega_{c_2}$ (Hz)        \\
    \hline   
           BPF&  1 TDC, 1 FI& $2\times 10^{-4}$ & 0 & 0.9, 0&60, 1 \\
      \hline   
           BPF&  1 TDC, 1 RC, 1 FI& $2\times 10^{-4}$ & estimated & 0.9, 1.05&60, 1 \\
\hline   
    \end{tabular}
  \end{center}
\end{table}

\begin{figure}
\centering
  \includegraphics[width=1\linewidth]{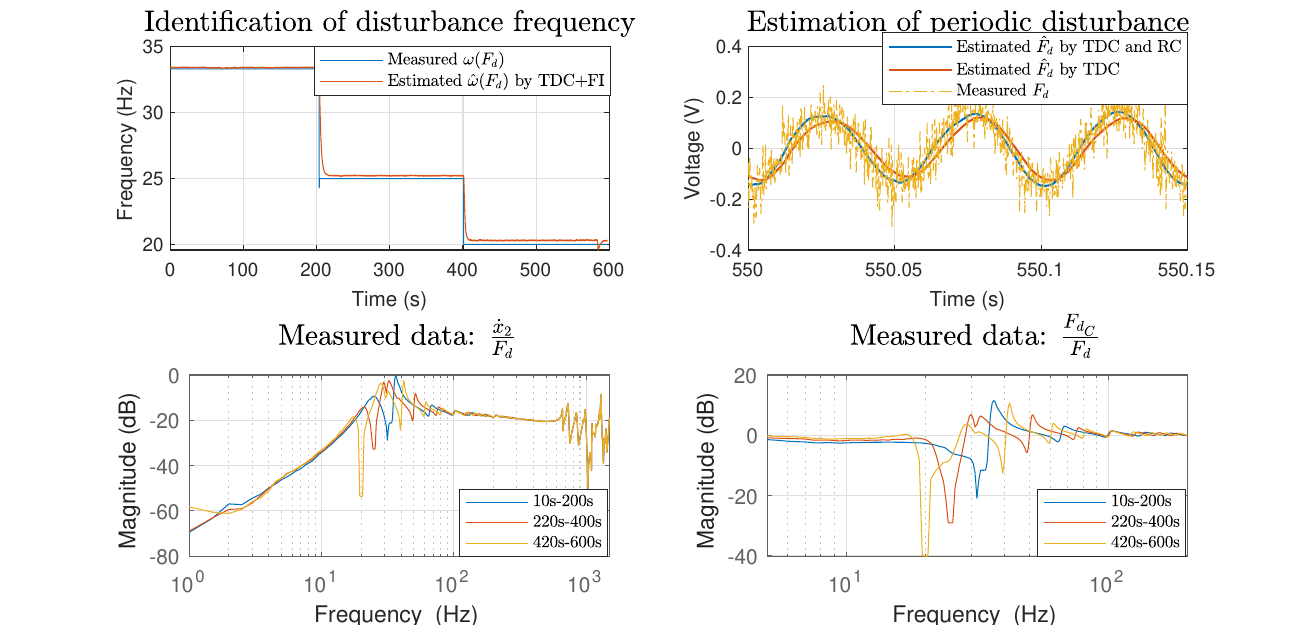}
  \caption{Experimental data: Tbe upper-left subfigure shows disturbance frequency identification using TDC and RC, with a frequency identifier in \protect\cite{Qing}. The upper-right subfigure shows the reconstructed periodic disturbance.
  The bottom-left subfigure shows the transfer function of acceleration to the applied noise disturbance. The bottom-right subfigure shows the compensated disturbance in each time duration, corresponding to their specific disturbance frequency. }
  \label{fig:Fd_iden}
\end{figure}

 Fig. \ref{fig:Fd_iden} shows the performance of disturbance estimation and frequency identification in the presence of noise in the disturbance. 
 The upper-right subfigure shows the reconstructed disturbance using the proposed approach for a disturbance frequency of 20 Hz, where the frequency estimation error is the largest. 
 Formulas used for reconstructing the disturbance estimate $\hat F_d$ by TDC with $n=1$, $\tau_1=2\times 10^{-4}$s via FI is given by $F_R$ in (\ref{F1_phi}) which is mathematically expressed as the equivalent disturbance in  (\ref{s_22}) as $\hat F_{d_{TDC}}=-F_R(n=1)$. 
 Formulas used for reconstructing the disturbance estimate $\hat F_d$ by TDC+RC with $n=2$ and $\tau_2$ obtained via FI is given by $\hat F_{d_{TDC+RC}}=-F_R(n=2)$. We have not proved the stability of this section, but $\tau_2$ needs to be bounded at least for stability. 
The disturbance can be reconstructed despite the noise in the disturbance, as the BPF in the TDC and the RC filters remove it.
The reconstructed disturbance $\hat F_{d_{TDC}}$ by TDC has a smaller amplitude than the actual disturbance because $\lambda_1=0.9$ is less than 1   and lags behind the actual disturbance. 
The phase lag is caused by both the sampling delay in the TDC and the phase lag from the BPF. 
However, the disturbance frequency can still be identified despite the phase lag.  

The upper-left subfigure shows the performance of the FI based on the disturbance reconstructed by TDC. The frequency estimation error increases for lower frequencies.
Since RC+TDC allows for increasing RC robustness to frequency uncertainties, as demonstrated in Fig. \ref{fig:3}, the performance of RC may not be severely affected in the presence of this level of error. 
The disturbance reconstructed by TDC and RC together in the upper-right subfigure then shows a more accurate match to the magnitude and phase of the actual disturbance, despite the estimation error being the largest at this frequency.   
This indicates that RC uses the disturbance frequency identified from the TDC-reconstructed disturbance, with a phase lag. Then it recovers the phase lag in the reconstructed disturbance as it is applied in addition to TDC.  
So, RC improves TDC by obtaining a more accurate estimate of the periodic disturbance, which RC then uses to compensate for the actual disturbance. 
 The bottom subfigures show the performance of TDC+RC  for each frequency of the varying disturbance frequency. 
The level of suppression due to RC+TDC reduces as the disturbance frequency increases, because of the constant RC gains and the filter cutoff frequency at 60 Hz. 
There is a flat region of suppression between 19.5 Hz and 20.5 Hz at -43.3 dB (the yellow line), and between 24.5 Hz and 25.5 Hz at -29.1 dB (the red line). This further confirms that TDC yields more robustness to RC. As a result, the combination of RC and TDC is robust against the error of the frequency identification to yield a more accurate disturbance reconstruction than TDC, as discussed.

\section{\textbf{Conclusion}}
We have shown that the RC and TDC disturbance observers have identical structures, except for the delays used. For RC, we choose the delay to be the period of the periodic disturbance. For TDC, we set the delay equal to the sampling period. By simply adding each single period RC in a nested fashion to the nominal control,
 the proposed approach can leverage both RC and TDC disturbance observers to allow multi-frequency RC,  as well as robust RC against frequency uncertainties. Comparison with a high-order 
 RC shows that our proposed approach increases robustness to frequency uncertainties while amplifying less at intermediate frequencies.

We show that the proposed MPRC has advantages over Parallel or Cascaded MPRC, as it does not require a high nominal control gain, as assumed by Cascaded MPRC, to reduce multi-frequency interaction. We experimentally scrutinized Cascaded and Parallel MPRC; the schemes restrict us from choosing larger nominal control gains before the system becomes unstable in our setup. 
The instability is purely due to the undesirable amplification of disturbance  at 20 Hz resulting from multi-frequency RC interaction, rather than parasitic dynamics at high frequencies in our setup.
 We proposed modifications to the Cascaded MPRC that do not require an identified plant model and show that this is identical to the proposed nested architecture. 

We present the closed-loop system in state-space form with multiple interacting delays and RC-induced disturbances. 
 We derived ultimate bounded solutions for the proposed MPRC framework with a low-pass or band-pass filter via input-to-state stability analysis. 
 This result can guide the tuning of the filters, RC, and nominal controller parameters to achieve predefined bounded performance at the designed frequencies.
 Our stability result accommodates optimization at both repetitive and non-repetitive frequencies. 
We conduct extensive experiments to verify the proposed framework's aforementioned functionalities. 
As part of the investigation, we also experimentally examined whether the proposed approach can leverage the features of TDC and RC to estimate and reject unknown periodic disturbances.

\appendix
\setcounter{equation}{0}
\renewcommand{\theequation}{A\arabic{equation}}

\subsection{State-space formulation for $n=3$}\label{Appen_A}
Following the  derivation for  $n=2$ case,  
we can extend the state-space formulation for $n=3$. The closed-loop system (\ref{s6_2}) becomes
\begin{equation}
\overline m\dot x_2=(1-Q_1(s)\lambda_1e^{-s\tau_1})(1-Q_2(s)\lambda_2e^{-s\tau_2})(1-Q_3(s)\lambda_3e^{-s\tau_3})   \Phi_1(s)-px_1-gx_2.
\end{equation}
Taking a Laplace expansion, we have
\begin{equation}
\begin{array}{l}
(s+\omega_{c_1})\big((s+\omega_{c_2})(s+\omega_{c_3})(ms^2+(d+g)s+k+p)X_1(s)-F_d(s)\big)\\
=\big(-\omega_{c_1}\lambda_1(s+\omega_{c_2})\\
\cdot (s+\omega_{c_3})e^{-s\tau_1}-\omega_{c_2}\lambda_2(s+\omega_{c_1})(s+\omega_{c_3})e^{-s\tau_2}-\omega_{c_3}\lambda_3(s+\omega_{c_1})(s+\omega_{c_2})e^{-s\tau_3}\\+\omega_{c_1}\omega_{c_2}\lambda_1\lambda_2(s+\omega_{c_3})e^{-s(\tau_1+\tau_2)}
+\omega_{c_1}\omega_{c_3}\lambda_1\lambda_3(s+\omega_{c_2})e^{-s(\tau_1+\tau_3)}\\+\omega_{c_2}\omega_{c_3}\lambda_2\lambda_3(s+\omega_{c_1})e^{-s(\tau_2+\tau_3)}\\-\omega_{c_1}\omega_{c_2}\omega_{c_3}\lambda_1\lambda_2\lambda_3e^{-s(\tau_1+\tau_2+\tau_3)}\big)\Phi_1(s).
\end{array}
\end{equation}
Denoting $D_m=\textrm{diag}\{1,\,1,\,1,\,1,\,m\} $, $x=\textrm{col}\{x_1,\,x_2,\,x_3,\,x_4,\,x_5\}  $ and \scalebox{0.8}{$F_{d_s}=\textrm{col}\{F_d,\,\dot F_d,\,\ddot F_d,\, \dddot F_d\}$},
the state-space form becomes
\begin{equation}\label{state 3 loop}
\resizebox{0.98\linewidth}{!}{$
\begin{array}{l}
D_m\dot x=Ax+A_1x(t-\tau_1)+A_2x(t-\tau_2)+A_3x(t-\tau_3)+A_4x(t-\tau_1-\tau_2)\\
+A_5x(t-\tau_1-\tau_3)+A_6x(t-\tau_2-\tau_3)+A_7x(t-\tau_1-\tau_2-\tau_3)\\
+Bx+B_1F_{d_s}(t-\tau_1)+B_2F_{d_s}(t-\tau_2)+B_3F_{d_s}(t-\tau_3)+B_4F_{d_s}(t-\tau_1-\tau_2)\\
+B_5F_{d_s}(t-\tau_1-\tau_3)+B_6F_{d_s}(t-\tau_2-\tau_3)+B_7F_{d_s}(t-\tau_1-\tau_2-\tau_3).
\end{array}
$}
\end{equation}
Defining $\bar A=A+\sum_{i=1}^7A_i$, $\bar B=B+\sum_{i=1}^7B_i$,
\begin{equation}
B_s=\begin{pmatrix}
\bar B&-B_1&-B_2&-B_3&-B_4&-B_5&-B_6&-B_7
\end{pmatrix}
\end{equation}
and
\begin{equation}
\begin{split}
\overline F_{d}=\textrm{col}\{&F_{d_s},\, \int_{t-\tau_1}^t\dot F_{d_s}(\mu)\textrm{d}\mu,\, \int_{t-\tau_2}^t\dot F_{d_s}(\mu)\textrm{d}\mu,\, \int_{t-\tau_3}^t\dot F_{d_s}(\mu)\textrm{d}\mu,\\
&\int_{t-\tau_1-\tau_2}^t\dot F_{d_s}(\mu)\textrm{d}\mu,\,
\int_{t-\tau_1-\tau_3}^t\dot F_{d_s}(\mu)\textrm{d}\mu,\,
\int_{t-\tau_2-\tau_3}^t\dot F_{d_s}(\mu)\textrm{d}\mu,\\
&\int_{t-\tau_1-\tau_2-\tau_3}^t\dot F_{d_s}(\mu)\textrm{d}\mu\}
\end{split}
\end{equation}
, 
the transformed system becomes
\begin{equation}\label{system}
\begin{array}{l}
D_m\dot x=\bar Ax-A_1\int_{t-\tau_1}^t\dot x(\mu)\textrm{d}\mu-A_2\int_{t-\tau_2}^t\dot x(\mu)\textrm{d}\mu-A_3\int_{t-\tau_3}^t\dot x(\mu)\textrm{d}\mu-A_4\int_{t-\tau_1-\tau_2}^t\dot x(\mu)\textrm{d}\mu\\-A_5\int_{t-\tau_1-\tau_3}^t\dot x(\mu)\textrm{d}\mu-A_6\int_{t-\tau_2-\tau_3}^t\dot x(\mu)\textrm{d}\mu-A_7\int_{t-\tau_1-\tau_2-\tau_3}^t\dot x(\mu)\textrm{d}\mu+B_s\overline F_{d}
\end{array}
\end{equation}
Following matrix parameters can be constructed as:
\begin{equation}
\resizebox{0.98\linewidth}{!}{$
\begin{array}{l}
\begin{pmatrix}
 A_{51}\\  \vdots \\  A_{55}
\end{pmatrix}=-\begin{pmatrix}\underbrace{\begin{pmatrix}
k&0&0&0\\d&k&0&0\\\Delta m&d&k&0\\0&\Delta m&d&k\\0&0&\Delta m&d
\end{pmatrix}}_K+\begin{pmatrix}
p&0&0&0\\g&p&0&0\\\bar m&g&p&0\\0&\bar m&g&p\\0&0&\bar m&g
\end{pmatrix}\end{pmatrix}
\underbrace{\begin{pmatrix}
\omega_{c_1}\omega_{c_2}\omega_{c_3}\\
\omega_{c_1}\omega_{c_2}+\omega_{c_1}\omega_{c_3}+\omega_{c_2}\omega_{c_3}\\ \omega_{c_1}+\omega_{c_2}+\omega_{c_3}\\1
\end{pmatrix}}_{\bar \omega}\\
\begin{pmatrix}
A_{1_{51}}\\ \vdots \\A_{1_{55}}
\end{pmatrix}
=\omega_{c_1}\lambda_1K(:,1:3)
\underbrace{
\begin{pmatrix}
\omega_{c_2}\omega_{c_3}\\ \omega_{c_2}+\omega_{c_3}\\1
\end{pmatrix}}_{\bar \omega_{23}}
,\,\begin{pmatrix}
A_{2_{51}}\\ \vdots \\A_{2_{55}}
\end{pmatrix}
=\omega_{c_2}\lambda_2K(:,1:3)
\underbrace{
\begin{pmatrix}
\omega_{c_1}\omega_{c_3}\\ \omega_{c_1}+\omega_{c_3}\\1
\end{pmatrix}}_{\bar \omega_{13}},
\\
\begin{pmatrix}
A_{3_{51}}\\ \vdots \\A_{3_{55}}
\end{pmatrix}
=\omega_{c_3}\lambda_3K(:,1:3)
\underbrace{
\begin{pmatrix}
\omega_{c_1}\omega_{c_2}\\ \omega_{c_1}+\omega_{c_2}\\1
\end{pmatrix}}_{\bar \omega_{12}},\,
\begin{pmatrix}
A_{4_{51}}\\ \vdots \\A_{4_{55}}
\end{pmatrix}
=-\omega_{c_1}\omega_{c_2}\lambda_1\lambda_2 K(:,1:2)\underbrace{\begin{pmatrix}
 \omega_{c_3}\\1
\end{pmatrix}}_{\bar \omega_3},\\
\begin{pmatrix}
A_{5_{51}}\\ \vdots \\A_{5_{55}}
\end{pmatrix}
=-\omega_{c_1}\omega_{c_3}\lambda_1\lambda_3 K(:,1:2)\underbrace{\begin{pmatrix}
 \omega_{c_2}\\1
\end{pmatrix}}_{\bar \omega_2},\,
\begin{pmatrix}
A_{6_{51}}\\ \vdots \\A_{6_{55}}
\end{pmatrix}
=-\omega_{c_2}\omega_{c_3}\lambda_2\lambda_3 K(:,1:2)\underbrace{\begin{pmatrix}
 \omega_{c_1}\\1
\end{pmatrix}}_{\bar \omega_1},
\end{array}
$}
\end{equation}
\begin{equation}
\resizebox{0.98\linewidth}{!}{$
\begin{array}{l}
\begin{pmatrix}
A_{7_{51}}\\ \vdots \\A_{7_{55}}
\end{pmatrix}
=\omega_{c_1}\omega_{c_2}\omega_{c_3}\lambda_1\lambda_2\lambda_3 K(:,1),\, 
\begin{pmatrix}
B_{51}\\ \vdots \\ B_{54}
\end{pmatrix}=\bar \omega,\, \begin{pmatrix}
B_{1_{51}}\\ \vdots \\ B_{1_{54}}
\end{pmatrix}=-\omega_{c_1}\lambda_1 I_4(:,1:3)\bar \omega_{23},\\
\begin{pmatrix}
B_{2_{51}}\\ \vdots \\ B_{2_{54}}
\end{pmatrix}=-\omega_{c_2}\lambda_2 I_4(:,1:3)\bar \omega_{13},\, \begin{pmatrix}
B_{3_{51}}\\ \vdots \\ B_{3_{54}}
\end{pmatrix}=-\omega_{c_3}\lambda_3 I_4(:,1:3)\bar \omega_{12},\, \begin{pmatrix}
B_{4_{51}}\\ \vdots \\ B_{4_{54}}
\end{pmatrix}=\omega_{c_1}\omega_{c_2}\lambda_1\lambda_2 I_4(:,1:2)\bar \omega_{3}\\
\begin{pmatrix}
B_{5_{51}}\\ \vdots \\ B_{5_{54}}
\end{pmatrix}=\omega_{c_1}\omega_{c_3}\lambda_1\lambda_3 I_4(:,1:2)\bar \omega_{2},\, 
\begin{pmatrix}
B_{6_{51}}\\ \vdots \\ B_{6_{54}}
\end{pmatrix}=\omega_{c_2}\omega_{c_3}\lambda_2\lambda_3 I_4(:,1:2)\bar \omega_{1},\\
\begin{pmatrix}
B_{7_{51}}\\ \vdots \\ B_{7_{54}}
\end{pmatrix}=-\omega_{c_1}\omega_{c_2}\omega_{c_3}\lambda_1\lambda_2\lambda_3 I_4(:,1).
\end{array}
$}
\end{equation}
We have just derived the matrices in (\ref{state 3 loop}). This procedure can be followed for  $n>3$.

\end{document}